\documentclass[a4paper,UKenglish,cleveref, autoref, thm-restate ]{lipics-v2021}

\usepackage{graphicx}

\usepackage{complexity}
\usepackage{mathtools}
\usepackage{tikz}
\usepackage{macros}
\usetikzlibrary{positioning}
\usetikzlibrary{decorations.markings}
\usetikzlibrary{matrix}
\usepackage{float}
\usepackage{algorithm, algpseudocode, float}
\usepackage{xspace}
\usepackage{graphicx}
\usepackage{amsmath}
\usepackage{amssymb}

\usepackage{hyperref}
\usepackage{cleveref}

\title{Parallel Complexity of Depth-First-Search and Maximal path in restricted
graph classes.} 
\author{Archit Chauhan\footnote{Part of this work was done when the author was
a PhD student at Chennai Mathematical Institute, India.}}{Department of Computer Science and Engineering, 
IIT Bombay, India}{archit@cse.iitb.ac.in}{}{}
\author{Samir Datta}{Chennai Mathematical Institute, India}{sdatta@cmi.ac.in}{}{}
\author{M. Praveen}{Chennai Mathematical Institute, India}{praveen@cmi.ac.in}{}{}

\authorrunning{A. Chauhan and S. Datta and M. Praveen} 

\Copyright{Archit Chauhan, Samir Datta, M. Praveen} 



\category{} 

\relatedversion{} 

\acknowledgements{We would like to thank the reviewers for their insightful 
suggestions and corrections which have significantly improved this paper.}

\nolinenumbers 

\begin{document}

\maketitle              
\begin{abstract}
Constructing a Depth First Search (DFS) tree is a fundamental graph problem, 
whose parallel complexity is still not settled. 
Reif showed parallel intractability of lex-first DFS. In contrast,
randomized parallel algorithms (and more recently, deterministic quasipolynomial 
parallel algorithms) are known for constructing a DFS tree in general (di)graphs. 
However a deterministic parallel algorithm for DFS in general graphs remains an elusive goal.
Working towards this, a series of works gave deterministic NC algorithms for DFS
in planar graphs and digraphs. We further extend these results to 
more general graph classes, by providing
NC algorithms for (di)graphs of bounded genus, and
for undirected $H$-minor-free graphs where $H$ is a fixed graph with at most one crossing. 
For the case of (di)graphs of bounded tree-width,
we further improve the complexity to a Logspace bound.\\
Constructing a maximal path is a simpler problem (that reduces to DFS) for which
no deterministic parallel bounds are known for general graphs. 
For planar graphs a bound of \bigo{log $n$} parallel time on a CRCW PRAM 
(thus in NC$^2$) is known. 
We improve this bound to Logspace.

\end{abstract}

\pagebreak

\section{Introduction}
Depth First Search (\textsf{\DFS}) in undirected graphs is an archetypal graph 
algorithm whose origin is in the 19-th century work of Tr\'emaux
\cite{EvenEven21}[Chapter 3]. 
The $\DFS$ algorithm has played a pivotal role in the 
study of several problems such as biconnectivity,
triconnectivity, topological sorting, 
strong connectivity, planarity testing and embedding etc. 
Along with Breadth First Search (\textsf{BFS}) it is perhaps one of
the best known algorithms. However, while the BFS algorithm is easy to 
parallelize, 
no such deterministic parallel algorithm is known for \DFS\footnote{When talking
  about parallel computation, we will mean the class \NC. 
  In terms of \textsf{PRAM} models, a problem is said to be in the \NC{} if there exists a 
\textsf{PRAM} machine and constants $k,c$ such that the machine solves the problem in
\bigo{log$^k n$} time, using \bigo{$n^c$} many parallel processors. If the
machine is a \textsf{C}oncurrent \textsf{R}ead \textsf{C}oncurrent
\textsf{W}rite PRAM, we denote this class by \textsf{CRCW}[$\log^kn$]. 
An equivalent definition in terms of boolean circuits also exists. In particular 
\textsf{CRCW}[$\log^kn$] is equivalent to $\AC^k$
(see~\cite[Section~3.4]{KR88}).}.
In fact, Reif showed in the 1980's~\cite{Reif85}
that unless all of $\P$ is parallelizable, there is no parallel algorithm
that can mimic the familiar recursive \DFS. To be slightly more precise 
finding the lexicographically first $\DFS$-order is $\P$-complete. Roughly
contemporaneously, the problem of finding \emph{a} (not necessarily lexicographically
first) $\DFS$-order was reduced to the minimum weight perfect
matching problem for bipartite graphs in \cite{AA88} 
(and in~\cite{AAK90} for directed graphs) which coupled 
with randomized parallel algorithm for min-weight 
perfect matching \cite{KUW86,MVV87}, yielded a randomized parallel algorithm for $\DFS$.
Recent advances in matching algorithms \cite{FGT19} indeed
yield a deterministic parallel algorithm for $\DFS$ although using 
quasipolynomially many (in particular, $n^{O(\textsf{log }n)}$ many) processors.

Faced with the quest of finding a deterministic parallel algorithm (with polynomially many
processors) for $\DFS$, 
a natural question is: for which 
restricted graph classes can we, in fact, parallelize $\DFS$. Planar graphs, 
as well as digraphs 
have been known to possess (deterministic) parallel $\DFS$ algorithms dating back to 
half a century ago~\cite{Smith86,Hagerup90,Kao88,KK93,Kao95,ACD22}. 
Khuller also gave an $\NC$ algorithm for $\DFS$ in $K_{3,3}$-minor-free
graphs~\cite{Khuller90}.
The above mentioned algorithms differ in their parallel
complexity. For example
Kao, Klein in~\cite{KK93} gave a parallel running time of $O(\log^{10}{n})$ for
$\DFS$ in planar digraphs while optimizing
the number of required processors to be linear. More recently, Ghaffari et al~\cite{GGQ23}
have given a (randomized) nearly work efficient algorithm for \DFS{} running in
$\tilde{O}$($\sqrt{n}$) time and $\tilde{O}$($m+n$) work.
Our metric for evaluating an algorithm is more complexity theoretic -- so 
we try to optimise on the complexity class we can place the problem in. 
In particular, we are interested in placing the problems in the smallest possible 
class in the following (loose) hierarchy : 
$\Log \subseteq \UL \subseteq \NL \subseteq \AC^1 \subseteq \NC^2 
\subseteq \AC^3 \subseteq \ldots$ (See Chapter 6 of~\cite{arora.barak}). 

Most algorithms for parallel \DFS{} in graphs proceed via finding 
a balanced path separator in the graph, which is a path such that removing all
of its vertices leaves behind connected components (or strongly connected 
components in case of digraphs) of small size, allowing for logarithmic depth
recursion. 
It was shown by Kao in~\cite{Kao88} that all graphs and digraphs have 
balanced path separators, and it is easy to find one sequentially using the
$\DFS$ algorithm itself.
The challenge thus is to find a path separator in parallel.
Another useful result Kao~\cite{Kao93,Kao95} also showed is that if
we have a small number of disjoint paths that together constitute a
balanced separator (by small we mean up to polylogarithmically 
many, we call these \emph{multi-path separators}), then we can trim and
merge them into a single balanced path separator in parallel.
If we could somehow merge and reduce a multi-path separator consisting of a large 
number of paths, say \bigo{$n$} many, to a multi-path separator consisting of up to 
constant (or even polylogarithmically many) number of paths 
in \NC, then we would immediately have an \NC{} algorithm for \DFS{}, as the set
of all vertices of the graph constitutes a trivial multi-path
separator. Such a routine to reduce the number of paths in a multi-path separator 
is indeed the heart of the \RNC{} algorithm of Aggarwal and
Anderson~\cite{AA88} (and also~\cite{AAK90}), and this is where the randomized
routine for minimum weight perfect matching in bipartite graphs is used.
We also remark that in the literature of separators, an aspect of focus is often 
to minimize the size of the separator~\cite{LT77,Miller86}.
For our purpose of constructing a DFS tree, the length of the path/cycle 
separator is not important. 

\subsection{Our contributions}
We first look at graphs of bounded genus in~\cref{sec:sep_surface}, 
which are a natural generalization of planar graphs (planar graphs have genus
zero), and show that \DFS{} is in \NC{} for such graphs. We are not aware of any  
deterministic $\NC$ algorithm for \DFS{} even for toroidal graphs. 
We give an \NC{} algorithm for $\DFS$ which proceeds via finding a multi-path separator consisting
of \bigo{$g$} many paths, where $g$ is the genus of the graph. 
For graph classes of higher genus, the main bottleneck we encounter is that of
testing their genus and finding an embedding into a surface of that genus in $\NC$.
More precisely, if we assume that we can find genus and embeddings of  
graphs of upto \bigo{log $n$} genus in \NC,   
then our algorithm still gives an \NC{} bound 
for graphs of up to \bigo{log $n$} genus. 

In~\cref{sec:sep_scmfree} we look at \DFS{} in single crossing minor free
graphs, which are another generalization of planar graphs (orthogonal to bounded
genus graphs). 
By Wagner's theorem~\cite{Wagner1937}, we know 
that planar graphs are precisely those whose set of forbidden minors consists of 
$K_5$ and $K_{3,3}$. Both $K_5,K_{3,3}$ are single crossing graphs, i.e. they can 
be drawn on the plane with at most one crossing. We consider the family of 
$H$-minor-free graphs, where $H$ is any fixed single crossing graph. 
Such families which include in particular the families of $K_5$-minor-free
graphs and $K_{3,3}$-minor-free graphs, are called single-crossing-minor-free
graphs. They form a natural stepping stone towards more general $H$-minor-free
families, which by a deep theorem of Robertson and Seymour~\cite{RS04},
completely characterize $\emph{minor-closed}$ families. 
The only result we know of in this direction is that of
Khuller's~\cite{Khuller90} parallel algorithm for \DFS{} in $K_{3,3}$-minor-free
graphs. 
We generalize this result by giving an NC algorithm for 
DFS in single-crossing-minor-free graphs.
The algorithm of~\cite{Khuller90} uses a decomposition of $K_{3,3}$-minor-free
graphs by $2$-clique sums (or separating pairs) 
into pieces that are either planar or are exactly the graph $K_5$. 
They proceed by finding the `center' piece of this decomposition tree and then 
finding a balanced cycle separator in it, after projecting the weights of
attached subgraphs on the center piece. The high level strategy of our algorithm is
similar. We first use a decomposition theorem~\cite{RS93} by Robertson and Seymour 
which states that single-crossing-minor-free graphs can be decomposed by 
\emph{$3$-clique sums} into pieces that are either planar or of bounded treewidth. 
Then we try to find a balanced cycle separator in the center piece.
However the attachment at $3$-cliques instead of $2$-cliques poses some
technical difficulties,  
as simply projecting the weights of attached subgraphs on the 
center piece does not work. 
To deal with this, we use appropriate gadgets while projecting, 
that mimic the connectivity of the attached subgraphs. 

We remark here that the problem of finding a minimum weight perfect matching 
was shown to be in (deterministic) $\NC$ for
bounded genus bipartite graphs~\cite{DKTV12,GST20} as well as for single-crossing-minor-free 
graphs~\cite{EV21}. 
However, the reductions in~\cite{AA88,AAK90} involve adding large bi-cliques in the 
input graph, which does not seem to preserve the genus, nor the forbidden 
minors of the graph. 

In~\cref{sec:DFS_btw} we look at bounded treewidth (di)graphs, for which we are not aware of 
any explicitly stated parallel algorithms for $\DFS$. We note however that the machinery 
of Kao~\cite{Kao88} easily yields an $\AC^1$($\Log$) (which is contained in
$\AC^2$. See~\cref{sec:prelims}.) algorithm for $\DFS$ in bounded 
treewidth graphs. We improve this bound to $\Log$.
Instead of using path separators, our technique for bounded treewidth (di)graphs involves 
the parallel (in fact, Logspace) version of
Courcelle's theorem from \cite{EJT10}. 
We also use an existing result~\cite{CF12,BD15} that states that we can 
add a linear ordering to the vertices of the graph without blowing up 
the treewidth of the structure. This combined with an observation of~\cite{TK95} 
characterizing the lex-first DFS tree in terms of lex-first paths to each vertex, 
allows us to express membership of an edge to the DFS tree, in $\mso$.

Finally, in~\cref{sec:maximal_path_undir} we turn to the problem of finding a
maximal path in a planar graph, or $\maxPath$. Given an input (undirected) graph
and a vertex $r$ in it, the 
problem is to find a (simple) path starting at $r$ that cannot be extended further, i.e., 
all the neighbours of its end point other than $r$ already lie on the path.
This problem easily reduces to $\DFS$ because every root to leaf path in a $\DFS$-tree
is a maximal path. However even for this problem, no deterministic parallel
algorithm for general graphs is known, and the $\RNC$ algorithm 
by~\cite{Anderson87} goes via reduction to perfect matching.
For restricted graph classes, the situation is clearer for $\maxPath$ as
compared to $\DFS$. 
Anderson and Mayr~\cite{AM87} gave a deterministic \bigo{log$^3 n$} time 
algorithm on a \textsf{CREW PRAM}, 
for graphs of bounded degeneracy~\footnote{A graph is said to be $k$-degenerate 
if every subgraph of it has a vertex of degree at most $k$.} These include for 
example, planar graphs, and even more generally, all $H$-minor free graph classes. 
In modern terms, their algorithm witnesses a bound of $\AC^1$($\Log$). 
A better bound of \CRCW{} for \maxPath{} in \emph{planar graphs} comes from 
the $\DFS$ algorithm of Hagerup~\cite{Hagerup90}, which again can be 
improved by a modern analysis to $\ULx$~\cite{ACDM19}. 
We improve this bound to $\Log$. 
Our algorithm proceeds by first reducing the problem to finding a maximal path  
in the triconnected components of a planar graph using the $2$-clique sum 
decomposition (also referred to as the \textsf{SPQR} decomposition) 
of~\cite{DLNTW22}. We then show that in the triconnected planar components, 
a maximal path can be found in the `outerplanar layers' of the graph.
The results we prove can be summarized as follows:
\begin{theorem} \label{thm:main}
We have the following bounds:
\begin{enumerate}
\item\label{it:genusDFS}
 $\DFS$ in bounded genus graphs and digraphs is in $\ACU$.
\item\label{it:SCMDFS}
 $\DFS$ in undirected single-crossing-minor-free graphs is in $\NC$.
\item\label{itm:btwDFS} 
 $\DFS$ in directed and undirected graphs of bounded treewidth is in $\Log$
\item \label{it:maximalUPlan}
$\maxPath$ in undirected planar graphs is in $\Log$.
\end{enumerate}
\end{theorem}

\section{Preliminaries}\label{sec:prelims}

We assume the reader is familiar with tree-decompositions, planarity, embeddings
on surfaces and graph minors.
When referring to concepts like treewidth, minors and genus of digraphs, we will always 
intend them to apply to the underlying undirected graph. 
In digraphs, we will use the term path to mean simple directed path unless otherwise
stated. We will use \emph{tail} of a path to refer to the starting vertex of the path 
and \emph{head} of a path to refer to the ending vertex of the path. 
The term \emph{disjoint paths} will always mean vertex disjoint paths, and the 
term \emph{internally disjoint paths} will refer to paths that are disjoint except 
possibly sharing end points. 
In a given graph $G=(V,E)$ and edge $(u,v) \in E(G)$, we denote the 
graph obtained by removing the edge $(u,v)$ (but keeping the vertices $u,v$ intact) 
by $G-(u,v)$.

We refer the reader to
~\cite{arora.barak} for definitions of classes 
$\Log, \NL$. 
The class $\AC^k$ is the class of problems for which there is a logspace-uniform
family of circuits $\{C_n\}$ consisting of AND, OR, and NOT gates, and each
circuit $C_n$ has depth \bigo{log$^kn$} and $n^{O(1)}$ gates.  
The class $\NC^k$ is defined similarly but with the additional restriction that
the gates in circuits have a fan-in of at most two.
A nondeterministic Turing machine is said to be
unambiguous if for every input, 
there is at most one accepting computation path. The class $\UL$ is the 
class of problems for which there exists an unambiguous Logspace Turing machine 
to decide the problem. In our setting, our outputs are not binary, but 
encodings of objects like a DFS tree, or a maximal path.  
Each of the classes mentioned above can be generalised to deal with 
\emph{functions} instead of just \emph{languages} in a natural way.
One may refer to~\cite{JK89} or~\cite[Section~2]{ACD22} for more details, specially what is 
meant by the computation of functions by a bounded space class like $\Log$,
$\NL$ or $\UL \cap \coUL$. 
In simplest terms, the output has a size of polynomially many bits
and each bit can be computed in the stated class.
We will often refer to a machine that computes a function as a transducer.
A constant number of $\Log$ transducers may be composed in $\Log$ (that is the
output of one, is the input of the next). A similar result holds for classes
like $\NL, \UL \cap \coUL$ (see \cite[Section~2]{ACD22} for a more elaborate explanation).

An important result in the Logspace world is that of Reingold~\cite{Reingold08}
which states that the problem of deciding reachability (as well as its 
search version) in undirected graphs is in \Log. 
Using this and the lemma of composition of
$\Log$ transducers, many routines like computing connected components left after
removing some vertices, or counting the number of vertices in a connected
component, can be shown to be in $\Log$. Another result we will use is that the
problem of reachability in directed planar graphs is in \ULx{}, as shown
by Bourke et al~\cite[Theorem~1.1]{BTV09}. This extends to graphs of bounded genus,
and in fact to graphs of \bigo{log$n$} genus if their polygonal schema is given
as part of input~\cite[Theorem~7]{GST19} (See also~\cite{KV10,TV12,DKTV12}). 
The theorem of~\cite{GST19} achieves this by giving a logspace algorithm to
compute \bigo{log $n$} size weight functions for edges that isolate shortest paths between
vertices of the input graph. In other words, for any pair of vertices $u,v$,
there is a unique path minimum weight path from $u$ to $v$ under the weight
function computed by~\cite[Theorem~7]{GST19}. Thierauf and Wagner~\cite{TW10}
showed that with
such path isolating weights assigned to edges, 
we can compute distances as well as shortest path between
vertices, in \ULx{}. These results will be useful in computing
arborescences\footnote{An arborescence of a digraph $G$ is a spanning tree
rooted at a vertex $r$, such that all edges of the tree are directed away from
$r$.} in bounded
genus graphs in~\cref{sec:sep_surface}.

Suppose $\mathcal{C}$ is a complexity
class (a set of boolean functions). We use $\AC^k(\mathcal{C})$ to denote
functions that can be computed by a family of $\AC^k$ circuits augmented with oracle
gates that can compute functions in $\mathcal{C}$. Note for example, that since 
$\ULx \subseteq \AC^1$, the class $\AC^1$(\ULx) is contained in $\AC^2$.




We will use the notion of Gaifman graph of a logical structure:
\begin{definition}
  Given a logical structure $\struct$, with a finite universe 
  $|\struct |$, 
and relation symbols which are interpreted as relations of $\struct$, 
the Gaifman graph of $\struct$, $\gaif{\struct}$ is a graph whose 
vertices are the elements of $|\struct |$, and for $u,v \in |\struct |$, we 
add an edge between them iff $u,v$ both occur in a relation tuple of $\struct$. 
\end{definition}
The treewidth of the structure is the treewidth of its Gaifman graph.
We will use the following extension of Courcelle's theorem by~\cite{EJT10}:
\begin{theorem}{\cite{EJT10}}\label{thm:EJT_C}
  For every $k \geq 1$ and every MSO-formula $\phi$, there is a logspace DTM
  that on input of any logical structure $\mathcal{A}$ of treewidth at most
  $k$, decides whether $\mathcal{A} \models \phi$ holds.
\end{theorem} 
\paragraph*{Clique-Sum decompositions}
We recall the notion of \emph{clique-sums}. Given two graphs $G_1,G_2$, a $k$-clique-sum 
of them can be obtained from the disjoint union of $G_1,G_2$ by identifying a clique 
in $G_1$ of at most $k$ vertices with a clique of the same number of vertices in $G_2$, 
and then possibly deleting some of the edges of the merged clique. 
It is used in reverse while decomposing a graph by clique sums. 
For our purpose, we will use up to $3$-clique sum decomposition of a
graph.
Suppose $G$ decomposes via a $3$-clique sum at clique $c$ into $G_1$ and $G_2$. Then we
write $G$ as $G_1{\oplus}_{c} G_2$.
More generally, if $G_1, G_2, \ldots, G_{\ell}$ all share a common clique $c$, then we
use $G_1\oplus_c G_2 \oplus_c \ldots \oplus_c G_{\ell}$
to mean $G_1, G_2, \ldots, G_{\ell}$ are glued together
at the shared clique.
If it is clear from the context which clique we are referring to, we will sometimes
drop the subscript and simply use $G_1\oplus G_2 \oplus \ldots \oplus G_{\ell}$ instead.
When separating the graph along a separating pair/triplet, we add virtual edges if 
needed to make the separating pair/triplet a clique. 
The decomposition can be thought of as a two colored tree
(see~\cite{CE13,EV21,CDGS24}) where the blue colored 
nodes represent $\emph{pieces}$ (subgraphs that the graph is decomposed into), 
and the red colored nodes represent the cliques at which two or more pieces 
may be attached. The edges of the clique describe the incidence relation between 
pieces and cliques (see~\cref{fig:clique_decomp}). 
For single-crossing-minor-free graphs, Robertson and Seymour gave the 
following theorem regarding their $3$-clique sum decomposition:
\begin{theorem}{\cite{RS93}}
 For any single-crossing graph $H$, there is an integer $\twh$ such that
 every graph with no minor isomorphic to $H$ can be obtained by
 $3$-clique sums, starting from planar graphs and graphs of
 treewidth at most $\twh$.
\end{theorem}

An observation we use is that we can assume that in any planar piece 
of the decomposition, the vertices of a separating pair or triplet lie 
on a common face (Otherwise we could decompose the graph further. 
See~\cite{CE13,EV21} for more detailed explanation).
It is shown in~\cite{EV21} that it can be computed in \NC. 

For computing the $2$-clique sum decomposition, 
we will use the Logspace 
algorithm given by~\cite{DLNTW22}. 
The resulting pieces of the decomposition, also called as the
\emph{triconnected components} of the graph, are either dipoles,
simple cycles, or non-trivial $3$-connected graphs (i.e. $3$-connected graphs 
with at least four vertices). 

\begin{figure}
\begin{minipage}{0.4\textwidth}
  \includegraphics[scale=0.62]{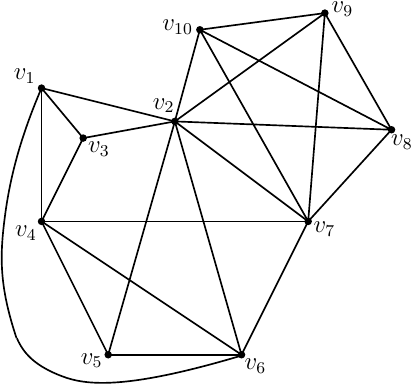}
  \caption{An example of a graph $G$. We ignore directions here.}
\end{minipage}
\begin{minipage}{0.6\textwidth}
  \includegraphics[scale=0.56]{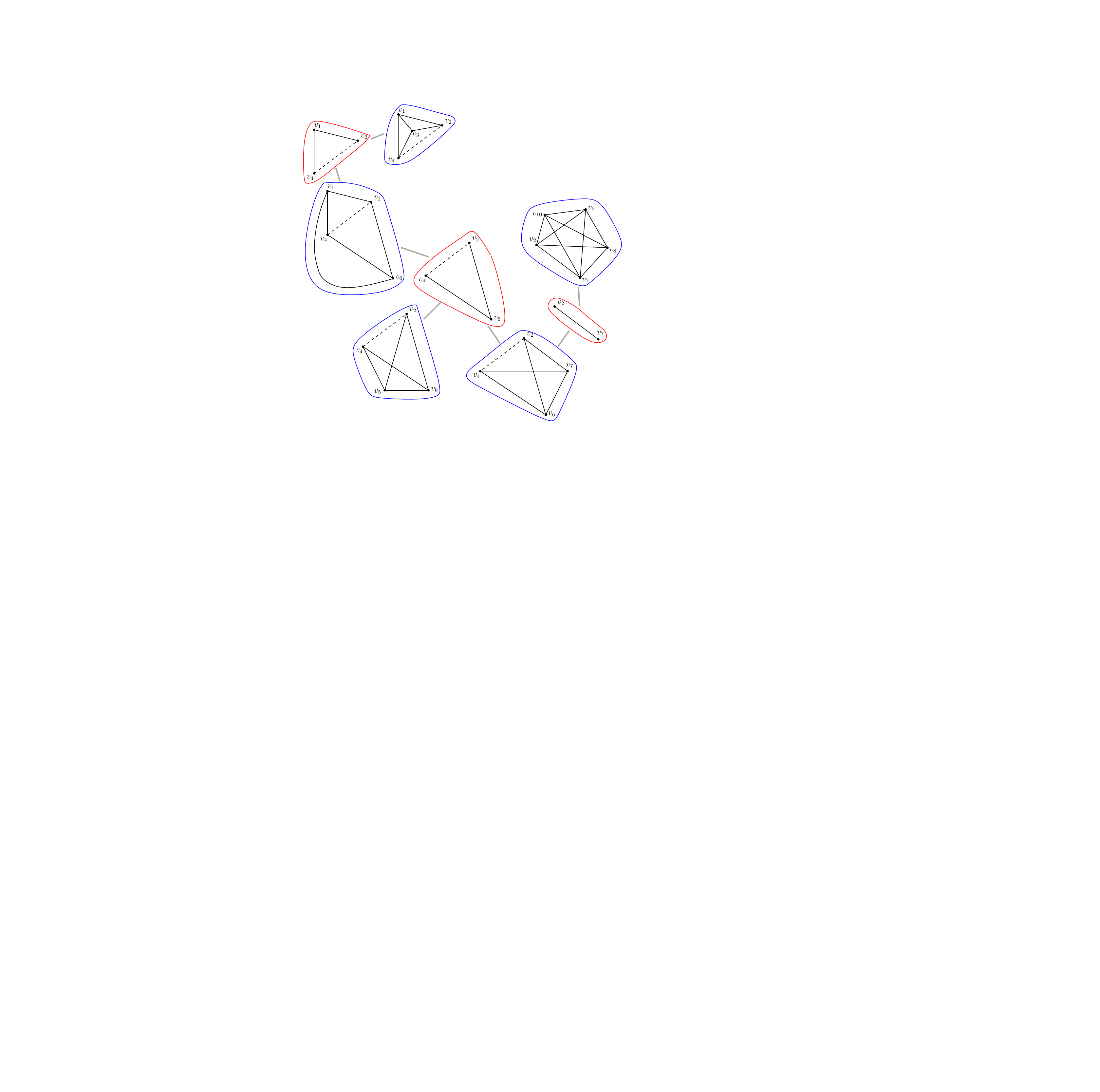}
  \caption{A clique sum decomposition of $G$. Red nodes
   are the clique nodes and blue node the piece nodes.
   Dashed edges denote virtual edges.}\label{fig:clique_decomp}
\end{minipage}
\end{figure}

\paragraph*{Balanced Path/Cycle separators}
Consider an input graph $G=(V,E)$ with $|V|=n$. 
If $G$ is undirected, we use the term $\alpha$-separator (or just 
balanced separator) in general to mean 
some subgraph of $G$ which if removed from $G$, then
the size of each of the resulting connected components is at most $\alpha n$, where 
$0 < \alpha < 1$. A single path $P$, or a collection of paths $\{P_1,\ldots P_k\}$
for example could form a balanced separator of $G$. We call the latter a
\emph{multi-path} separator consisting of $k$ paths. 
If $G$ is directed, we have a similar definition, but the constraint is instead on the size
of each \emph{strongly} connected component left after removing the vertices 
of the separator.
Note that if we have an algorithm to find an \sep{\alpha} in a directed graph, then
we can use that to find an \sep{\alpha} in an undirected graph by making each edge
into a bi-directed edge. 
Any simple cycle $C$ in an embedded planar graph $G$ will divide it into two regions, the
\emph{interior} and \emph{exterior} of $C$, which we denote by $\inte{C}$ and
$\exte{C}$ respectively.
Typically, if a face of $G$ is marked as the outer face then
the region containing the outer face is referred to as the exterior and the
other region as the interior.
Suppose the vertices, edges and faces of $G$ are assigned weights. Then we use
$\wt(\inte{C})$ to denote the sum of weights of vertices, edges, faces that are
contained in $\inte{C}$. The term $\wt(\exte{C})$ is defined similarly. The vertices,
edges that lie on $C$ itself are usually not
counted since if they are part of the separator and will be removed when
recursing. 
Suppose $u,v$ are two vertices in the cycle $C$ (embedded in the plane). 
We denote by $\segr{u}{v}$ (respectively $\segl{u}{v}$),
the segment of $C$ obtained by starting from $u$ and going clockwise (respectively
counter-clockwise) along $C$ till $v$.

%
%
%
%

\section{Path/Cycle Separators review}\label{sec:path_separators}

We restate some of the existing algorithms and analyze their complexity with respect
to circuits/space bounded complexity classes.

First, we reiterate an algorithm of Kao (see also~\cite{Kao88,AAK90,Kao95}) that combines a 
small number of path separators into a single path separator in parallel, and observe 
that it is in $\Log$ given an oracle for $\reach$. 
We state the theorems and corollaries for digraphs, but the corresponding
statements for undirected graphs also hold with appropriate changes.

\begin{restatable}{theorem}{thmKao}(\cite{Kao93,Kao95}) \label{thm:Kao}
Let $G$ be a digraph and $\{P_1,P_2\}$ be two paths that together form an
$\alpha$-separator of $G$ for some $\alpha \in [\frac{1}{2},1)$. Given an
oracle for $\reach$, we can in \Log, construct a single path $P_3$ that 
is also an $\alpha$-separator of $G$.
\end{restatable}
We reproduce the proof for completeness in~\cref{app:thm_Kao_proof}. 
We demonstrate it for the case when $\alpha = \frac{1}{2}$, as
the argument is identical for higher values of $\alpha$.
Note that as a corollary, the above algorithm is in 
$\NL$ for general directed graphs, and
in $\ULx$ for planar directed graphs (using~\cite{BTV09}).
For undirected graphs, the algorithm (with appropriate changes) is in $\L$
using Reingold's algorithm~\cite{Reingold08}.
As noted in \cite{AAK90,KK93}, if we have a multipath separator 
consisting of $k$ many paths, say $\{P_1,P_2\ldots P_k\}$, 
we can combine them into a single path separator using $k$ iterations of 
the above procedure with $k$ transducers. 
Thus we have the following corollary:
\begin{corollary}\label{corr:Kao_ext}
Let $G$ be a digraph and $k$ be a constant. Let $\{P_1,P_2, \ldots P_k\}$ be $k$ disjoint paths 
that form a multipath $\alpha$-separator of $G$ for some $\alpha \in [\frac{1}{2},1)$. 
Given an oracle for $\reach$, we can in \Log, construct a single path $P_{k+1}$ that 
is also an $\alpha$-separator of $G$.
\end{corollary}

  
For bounded treewidth graphs, we can find a tree decomposition in $\Log$
using~\cite[Theorem~1.1]{EJT10}, and we know that there is always a bag in the
tree decomposition of a graph whose vertices together form a $1/2$-separator of
the graph.
Since $\reach$ is known to be in $\Log$ for digraphs of bounded treewidth 
(from~\cref{thm:EJT_C}), we have the following result using the
corollary above: \begin{corollary}
  Given a digraph of bounded treewidth, we can find a 
  $1/2$-path separator in $\Log$.
\end{corollary}

\subsection{Application of Path Separators in Depth First Search}
It is shown in~\cite[Section~7]{ACD22} how to construct a DFS tree in a digraph given
a routine for finding path separators, in parallel. The algorithm follows
a divide and conquer strategy, where we find DFS trees on the smaller
strongly connected components in parallel, and sew them together using a 
routine for DFS in DAGs. 
We state the result in the following proposition.
\begin{proposition}\label{prop:sep_to_DFS} 
  Let $G$ be a digraph and $r \in V(G)$. 
  Let $\fnsep(H)$ be a function that takes as input any digraph $H$  
  and returns a balanced path separator of $H$. 
  Then there exists a uniform family of $\AC^1$ circuits with oracle gates for
  $\fnsep$ and functions computable in \ULx, that take as input $(G,r)$, 
  and return a DFS tree of $G$ rooted at $r$.
\end{proposition}

Although the proposition is for
digraphs, the same idea works in the simpler setting of undirected graphs also
(Instead of DAGs of strongly connected components, we just have to recurse on
connected components there. See for example~\cite{Smith86}). 

\subsection{Cycle Separators in undirected planar graphs }
Consider the following, more general notion of cycle separators in 
planar graphs, that was introduced in~\cite{Miller86}:
\begin{definition}
 Let $G$ be an undirected planar graph, and $\wt\colon \{V(G)\cup E(G) \cup
 F(G)\}\to\Q^+$ be a normalised weight function assigning weights to
 vertices, edges and faces of $G$ that sum up to $1$. 
 We call a cycle $C$ a $2/3$- 
 $\emph{interior-exterior}$ cycle separator (or just a 2/3-cycle separator
 for brevity) if $\wt(\inte{C}) \leq 2/3$, and $\wt(\exte{C}) \leq 2/3$.
\end{definition}
Miller showed in~\cite{Miller86} that such a weighted cycle separator always exists
if $G$ is biconnected and if no face is too heavy. 
Moreover, we can find such 
a separator in parallel (see also~\cite{Shannon88}).

\begin{theorem}{\cite{Miller86}}\label{thm:IE_cyc_sep}
  Let $G$ be a biconnected planar graph, and $\wt\colon \{V(G)\cup E(G) \cup
 F(G)\}\to\Q^+$ be an assignment
  of non-negative weights to vertices, edges, faces that sum up to $1$, such that no face
  has weight more than $2/3$. Then we can construct a $2/3$-\ie-cycle separator of
  $G$ in \CRCW.
\end{theorem}
The theorem in~\cite{Miller86} also gives a bound on size of the separator but
that is not important for us.

\section{DFS in directed graphs of bounded genus.}
\label{sec:sep_surface}

We show how to find balanced path separators (and therefore a \DFS{} tree) 
in digraphs that can be embedded on a surface of genus at most a fixed constant 
$g$, in \NC{}. 
The first step naturally is to find an embedding, which can be done
using the result of Elberfeld and Kawarabayashi~\cite[Theorem~1]{EK14}, 
which states that we can find an embedding of a graph of genus at most a
constant $g$, on a surface of genus $g$, in Logspace.

The idea is to find a small number of disjoint, 
surface non-separating cycles (not necessarily directed cycles), say
$C_1, C_2, \ldots C_l$ ($l \leq g$) in $G$,
such that the (weakly) connected components left after removing them are all 
planar. These can be found with at most $g$ sequential iterations. 
Then, if any of the remaining planar components has a large strongly connected
component, we can find a $2/3$-
path separator, say $P$, in it using
~\cite[Theorem~19]{KK93}. 
A more modern analysis of the algorithm in $\cite{KK93}$ gives a $\ULx$ bound
(see~\cite[Section~3.3]{Chauhan24} for a proof). 
Then $C_1, C_2 \ldots C_l, P$ together form a balanced multipath separator 
of the digraph $G$.
The (possibly undirected) cycles $C_1, C_2 \ldots C_l$ will have an additional property. 
Any $C_i \in \{C_1, C_2 \ldots C_l\}$ will either be a 
directed cycle, or it will consist of two internally disjoint (directed) 
paths that share a common tail vertex and a common head vertex.
Therefore the multipath separator can be written as consisting of 
at most $2g+1$ many paths, and we can 
use~\cref{corr:Kao_ext} to get the required path separator.

We use the following known fact (see~\cite[Lemma~4.2.4]{MT01}
or~\cite[Fact~9.1.6]{Grohe17}).
\begin{theorem}{\cite{MT01}}
  Let $G$ be a graph embeddable on a surface $\sur$, of genus $g$. 
  Let $C$ be a cycle that forms a surface non-separating curve on some
  embedding of $G$ on $\sur$. Then every connected component of $G-C$
  is embeddable on a surface of genus $g-1$. 
\end{theorem}

Therefore, if we have an algorithm to find such a (undirected) cycle in $\NC$, then
we can repeat it $g$ times to find the required (undirected) cycles 
$C_1,C_2 \ldots C_k$. 

We will use the following lemma from \cite{ADR05} (See also \cite{Thomassen90}) 
to find such cycles.

\begin{lemma}{\cite{ADR05}}\label{lem:fund_cyc}
  Let $G$ be a graph of genus $g>0$, and let $T$ be a spanning tree of $G$. Then
  there is an edge $e \in E(G)$ such that $T\cup {e}$ contains a surface
  non-separating cycle. 
\end{lemma}

Clearly the theorem also works if we have an arborescence instead of
a spanning tree. The fundamental cycle $C$ corresponding to a non-tree edge
$e=(u,v)$ would consist of the non-tree edge $e$ and two vertex disjoint 
directed paths, say  $P,P'$, starting from a common ancestor of $(u,v)$ in the 
arborescence (we can include the common ancestor vertex in either of $P$ or
$P'$), 
and ending at $u,v$ respectively.

We now describe the steps of the algorithm for computing a path separator in a 
bounded genus digraph $G$:
\begin{enumerate}
  \item Decompose the graph into strongly connected components. If there
    is no strongly connected component of size larger than $2n/3$,
    then we are done, else let $G_0$ be the component larger than $2n/3$.
  \item Find an embedding of $G_0$ using~\cite{EK14}.
  \item Find an arborescence of $G_0$. To do this, first construct the path
    isolating weigths for edges using~\cite[Theorem~7]{GST19} so that there is a
    unique min weight path between any two vertices of $G_0$.
    Then for every vertex, we find the
    shortest path from a vertex say $r_0$ to it in \ULx{} using~\cite{TW10}
    (see~\cref{sec:prelims}). This will clearly give an arborescence of $G_0$
    rooted at $r_0$.
  \item By~\cref{lem:fund_cyc}, we known that there exists an edge such that its
    fundamental cycle, is a surface non-separating cycle. To find it we
    can go over every edge, compute its fundamental cycle $C$ and use the theorem 
    of~\cite{EK14} to check if every weakly connected component of $G_0-C_1$ has genus 
    strictly lesser than that of $G_0$. Let $C_1$ denote the surface
    non-separating cycle we found. This can be done in Logspace using
    transducers. 
  \item Let $G_1$ be the largest strongly connected component of
    $G_0-C_1$.
    If it is more than $2n/3$ in size then we do the steps above on $G_1$.
    Repeat until finally we have graph $G_{rem} = G - (C_1\cup C_2 \ldots C_l)$
    (where $l<g$) such that
    either all strongly connected components of $G_{rem}$ are smaller than
    $2n/3$ or all the weakly connected components are planar. In the latter
    case, find a $2/3$-path separator $P$ in
    the \scc{} in $\Grem$ that is of size more than $2n/3$.
  \item 
    For each cycle $C_i$ we found at each step, let $P_i,P'_i$ denote the two
    directed paths that form the cycle $C_i$ as explained above. The paths $P_1
    \cup P'_1 \ldots P_l
    \cup P'_l \cup P$ ($P$ if required) together form a $2/3$-separator of $G$.
  \item Use~\cref{corr:Kao_ext} to merge the paths $P_1, P'_1, \ldots P_l,P'_l,P$
    into a single path $P_s$ which is also a $2/3$-path separator of $G$. Output
    $P_s$.
\end{enumerate}

Steps 1 to 4 can be done in \Log{} or in \ULx. Step 5 involves at most $g$ iterations
of steps 1 to 4 and is therefore in \ULx{} as the class is closed under
composition for constant number of transducers. 
The merging of at most
$2l+1$ paths into a single path separator is in \ULx{} by~\cref{corr:Kao_ext}. 
Hence the entire algorithm to construct the separator is in \ULx{} as it uses
a constant number of transducers.

Therefore we can find a $2/3$-path separator in directed graphs of bounded
genus in \ULx. Using~\cref{prop:sep_to_DFS}, we can find a depth
first search tree in such graphs in \AC$^1$(\ULx), finishing the proof of part 1
of~\cref{thm:main}. 


\section{Path Separators in Single-Crossing-Minor-Free Graphs}\label{sec:sep_scmfree}

We start with computing the $3$-clique-sum decomposition by the 
$\NC$ algorithm of~\cite{EV21}.
We will denote this decomposition tree by $\T$, each node of which
is either a piece node or a clique node, and the weight of each
node is the number of vertices in the corresponding piece/clique.
(Note that by this convention, the sum of weights of all nodes of $\T$
will sum up to more than $n$ since some vertices would occur in multiple pieces).

\begin{definition} 
 Let $\pnode{\Tc}$ be a node of $\T$, and $\Tc$ its corresponding piece/clique, 
 such that the size of any connected component in $G-\Tc$ 
 is less than $n/2$. We call $\pnode{\Tc}$ a \cent{} node of $\T$, and $\Tc$ 
 a \cent{} piece/clique.
\end{definition}

Such a node always exists and can be found in Logspace by traversing along the
heaviest child in $\T$ (see Chapter 7 of~\cite{parameterized.algorithms} for
example). 

We will assume $\pnode{\Tc}$ to be the root of $\T$ hereafter.
We will denote the \emph{subtrees} of $\T$ that are children of $\pnode{\Tc}$
by $\{T_1,T_2, \ldots T_l\}$, the subgraphs of $G$ corresponding to these 
subtrees by $\{G_1,G_2, \ldots G_l\}$, and the cliques by which they are 
attached to $\Tc$ by $\{c_1,c_2 \ldots c_l\}$ respectively. These subgraphs 
might themselves consist of smaller subgraphs glued at the common clique. 
For example, $G_1$ might consist of $G_{11}, G_{12}, \ldots G_{1\indd}$ 
that are glued at the shared clique $c_1$ 
(i.e. $G_1 = G_{11}\oplus_{c_1} G_{12} \ldots \oplus_{c_1} G_{1\indd}$).
See~\cref{fig:scm_center} for reference. Note that because $\Tc$ is a \cent{} node, 
$\wt(G_{1i})- \wt(c_1) < n/2$ $\forall i\in [1..\indd]$.

If $\pnode{\Tc}$ is a clique node, then we have a $1/2$-separator of
at most three vertices and we are done. 
Therefore there are two cases to consider. 
Either $\Tc$ is a planar piece, or a piece of bounded treewidth.

\paragraph*{When $\Tc$ is of bounded treewidth}
In this case, we will refine the tree $\T$ by further decomposing $\Tc$.
Since $\Tc$ is of treewidth at most $\twh$, we can compute a tree decomposition 
of $\Tc$ such that every bag is of size at most $\twh +1$, in \Log{} using~\cite{EJT10}.
Let this tree be denoted by $T_{\Tc}$. 
Consider the subtree $T_1$ of $\T$ as described above, which has a node 
$\pnode{c_1}$ by which it is attached 
to $\Tc$. Since $c_1$ is a clique, 
there must be at least one bag in $T_{\Tc}$ that 
contains all vertices of $c_1$ (see~\cite{Bodlaender89}). 
Attach $\pnode{c_1}$ to any such node of $T_{\Tc}$. 
Do this procedure for all subtrees $T_1,T_2,\ldots T_{l} $. 
It is easy to see that 
this will result in another tree decomposition, and that at least one of 
the nodes of $T_{\Tc}$ will be a \cent{} node of this new tree. 
Hence we can use the procedure described above to find it in $\Log$. 
Since the bags of $T_{\Tc}$ are of size at most $\twh +1$, 
we get a $\sep{1/2}$ of $G$ consisting of constant number of vertices, and we are 
done.
\begin{figure}
\begin{minipage}{\textwidth}
  \includegraphics[scale=0.47]{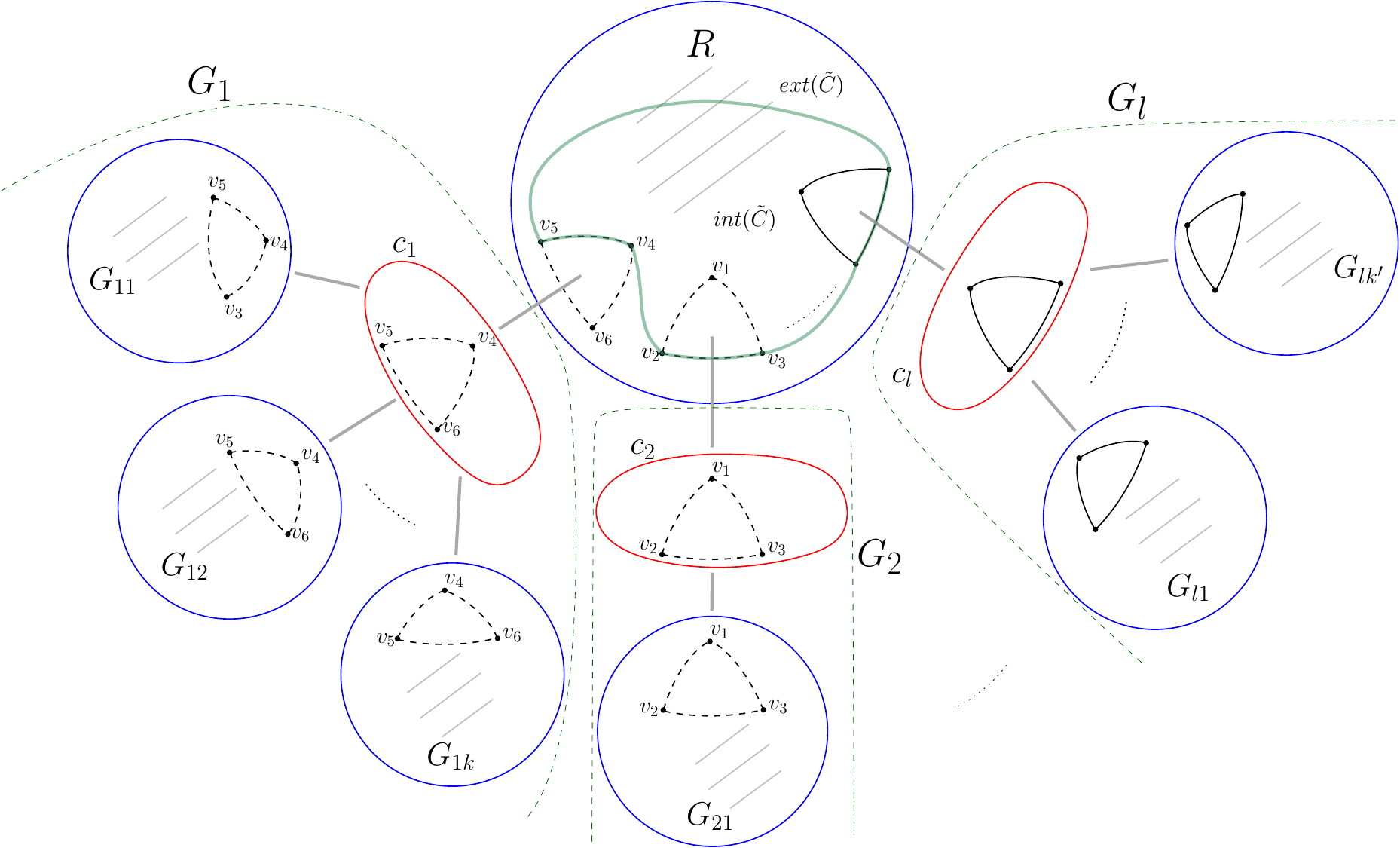}
  \caption{$\Tc$ is the central node with children subgraphs 
    $G_1,G_2,\ldots G_l$ attached to it via cliques $c_1,c_2,\ldots c_l$ respectively. 
    $G_1$ consists of subgraphs $G_{11},G_{12}\ldots G_{1k}$ glued at $c_1$. 
    $G_2$ is same as $G_{21}$.  Dashed edges denote virtual edges in $\Tc$.
    In the case when $\Tc$ is a planar piece, it is useful to think of it as a
    sphere, with separating triplets forming empty triangular faces in $R$, 
    on which $G_1,G_2,\ldots G_l$ are attached.
    The cycle in green is an example of \ie-cycle separator $\clos{C}$ (though gadget
    replacements are not shown in the figure). Of the attached subgraphs shown, 
    $G_2,G_l$ are \emph{interior components} with respect to $\tilde{C}$ whereas 
    $G_1$ is an \emph{exterior component} with respect to $\tilde{C}$. 
  }\label{fig:scm_center} 
\end{minipage} 
\end{figure}

\paragraph*{When $\Tc$ is planar}
We start with a simple idea, but encounter an obstacle. 
Then we show how to fix it.
Note that $\Tc$ (along with its virtual edges), is biconnected. 

To initialize, we assign weight $1$ to all vertices of $\Tc$ and
weight $0$ to all edges and faces of $\Tc$.
We project the weight of subgraphs $G_1,G_2 \ldots G_{\inda}$ on vertices, 
edges and faces of $\Tc$, at the respective cliques they are attached at.
Suppose the subgraph $G_k$ is attached to $R$ at clique $c_k$. Then we project
the weight as follows.
\begin{enumerate}
 \item If $c_k$ is a single vertex $v$, then assign a weight
   equal to $|V(G_k)|$ to \emph{vertex} $v$.
 \item If $c_k$ is a separating pair $(u,v)$, then we
   assign a weight equal to $|V(G_k)|-2$ to the \emph{edge}
  $(u,v)$ of $\Tc$. (We subtract $2$ since weights of vertices $u,v$ 
    are already accounted for).
 \item If $c_k$ is $3$-clique of vertices $(v_1,v_2,v_3)$, then
   we assign weight equal to $|V(G_k)|-3$ to the \emph{face} of 
   $\Tc$ enclosed by $c_k$.\label{item:3clique_proj}
\end{enumerate}
Let $\clos{\Tc}$ denote the weighted version of $\Tc$ obtained after projecting 
these weights. The sum of weights of all vertices, edges, faces of $\Tc$ is $n$. 
There are two possibilities. Suppose there exists a clique, say $c_1$ 
(which is either a vertex, or an edge, or a face of $\clos{\Tc}$), with weight at least 
$2n/3$. This means that total weight of subgraphs attached to $\Tc$ via $c_1$, 
i.e $\wt(G_{11}\oplus_{c_1} G_{12} \oplus_{c_1} \ldots \oplus G_{1\indd})$, is 
at least $2n/3$. Since each of the graphs $G_{11}, G_{12} \ldots G_{1\indd}$ is 
of size at most $n/2$ and they all get disconnected on removal of $c_1$, the 
vertices of $c_1$ form a $\sep{1/2}$ of $G$, and we are done. 

Therefore the remaining possibility we need to consider is one where 
every vertex, edge, face of $\clos{\Tc}$ has weight smaller than $2n/3$.
Since $\clos{\Tc}$ is planar and biconnected, we can find a 
$2/3$-\ie-cycle-separator of it in $\Log$ using~\cref{thm:IE_cyc_sep}. 
Let $\clos{C}$ be such a separator. Label one of the regions $\clos{C}$ divides 
$\Tc$ into as the interior of $\clos{C}$ and the other as the exterior 
of $\clos{C}$. Observe that for any subgraph that 
is attached to $\Tc$ (like $G_{2}$ for example), its vertices of attachment to 
$\Tc$ (the clique via which it is attached), will either all lie on or in the interior 
of $\clos{C}$, or all lie on or in the
exterior of $\clos{C}$. We call the attached subgraphs 
of the former kind \emph{interior components} of $\Tc$ with respect to $\clos{C}$, 
and the latter as \emph{exterior components} of $\Tc$ with respect to $\clos{C}$
(see~\cref{fig:scm_center}). We will drop the phrase `with respect to
$\clos{C}$' for brevity.  On removing the vertices of $\clos{C}$ in $G$, 
there cannot remain any path in $G$ from any vertex of an interior component 
of $\Tc$ to any vertex of an exterior component of $\Tc$. 
Therefore we have the following claim:
\begin{claim}\label{claim_sep_lift}
  Let $\clos{C}$ be an $\alpha$-\ie-cycle-separator of $\clos{\Tc}$, where
  $\alpha \in [\frac{1}{2},1)$. Then the 
  vertices of $\clos{C}$ form an \sep{\alpha} of $G$.
\end{claim}
Now if we could replace the virtual edges of $\clos{C}$ by valid path segments 
in the attached subgraphs 
then we would get a simple cycle $C$ in $G$ 
which would be a $\sep{2/3}$ of $G$, 
and we would be done. 
However there is an issue here, as the virtual edges in a
$3$-clique need not capture \emph{``disjointness''} of the external connections
accurately. For example,
suppose $c_2 = \{v_1,v_2,v_3\}$ is a $3$-clique and the subgraph 
$G_{2}$ is attached to $\Tc$ at $c_2$ (as shown
in~\cref{fig:scm_center}). Suppose $(v_1,v_2)$, and 
$(v_1,v_3)$ are virtual edges in $\Tc$, and our cycle separator $\clos{C}$ uses both of these 
edges (they must necessarily be consecutive in $\clos{C}$). Though there 
must be paths between $v_1,v_2$, and between $v_1,v_3$ in $G_{2}$ (since it 
is attached via $3$-clique), there might not exist such paths that are internally 
disjoint. 
Hence instead of just projecting the weight of $G_{2}$ on the face formed by
$c_2$ in $\Tc$, we use a more appropriate gadget to mimic the connectivity 
of $G_{2}$. 

\paragraph*{Gadgets for the subgraphs attached to $\Tc$ via $3$-cliques.}

The gadget for an attached graph like $G_2$ is essentially a coarser planar 
projection of its block-cut tree, that captures the needed information
regarding the cut vertices associated with its terminals $\{v_1,v_2,v_3\}$.
We can adversarially assume that all 
three edges $(v_1,v_2),(v_2,v_3),(v_3,v_1)$ are virtual edges.
\begin{definition}
  Let $G_{2}$ be a graph and vertices $v_1,v_2,v_3 \in V(G_{2})$ be 
  called its terminals.
The \emph{disjoint path configuration} of $G_{2}$ with respect to its terminals 
$v_1,v_2,v_3$ consists of three boolean variables : 
$\dpc{v_1}{v_2}{v_3}, \dpc{v_2}{v_3}{v_1}, \dpc{v_3}{v_1}{v_2}$.
$\dpc{v_1}{v_2}{v_3}$ takes value \tru{} if there exists a path between 
$v_1$ and $v_3$ that goes via $v_2$, and \fal{} otherwise.
$\dpc{v_2}{v_3}{v_1}, \dpc{v_3}{v_1}{v_2}$ are defined similarly.
\end{definition}
Their are four cases we consider for possible disjoint path configurations, 
others are handled by symmetry. 

\textbf{Case 1 : }
\begin{tabular}{c c c}
  $\dpc{v_1}{v_2}{v_3}$ & $\dpc{v_2}{v_3}{v_1}$ & $\dpc{v_3}{v_1}{v_2}$ \\ 
   \tru & \tru & \tru
\end{tabular}

This means $G_2$ is biconnected, and hence 
we do not need a new gadget. We just keep the virtual edges 
between the clique vertices and assign weight equal to $|V(G_{2})|-3$ to 
the face enclosed by them (as described in~\cref{item:3clique_proj}). 
If $\clos{C}$ takes two virtual edges of $c_2$, 
we can use the two disjoint path algorithm 
in~\cite{KMV92} to find the corresponding paths in $G_{2}$. 

\textbf{Case 2 : }
\begin{tabular}{c c c}
  $\dpc{v_1}{v_2}{v_3}$ & $\dpc{v_2}{v_3}{v_1}$ & $\dpc{v_3}{v_1}{v_2}$ \\ 
   \tru & \tru & \fal
\end{tabular}

By Menger's theorem (see~\cite{Menger1927}, or Chapter 3 of~\cite{Diestel}), 
we know that there must be at least one cut 
vertex in $G_{2}$ that separates $\{v_1\}$ from $\{v_2,v_3\}$. Let $x$ be a 
cut vertex separating these. Let the connected component of $G_{2}-x $
containing $v_1$, augmented with $x$, be $H_1$. Let the component containing
$v_2,v_3$, augmented with $x$, be $H_0$ (i.e. $V(H_1)\cap V(H_0)=\{x\}$ and 
$V(H_1)\cup V(H_0)=V(G_{2})$). 
We choose $x$ such that there is no cut vertex in $H_0$, that separates 
$\{v_1\}$ and $\{v_2,v_3\}$ in $G_{2}$ (i.e. $x$ is the cut vertex described that is
`furthest' from $v_1$).
It is easy to see that $x$ is unique. 
By our choice of $x$, there 
must be a path in $H_0$ from $v_2$ to $v_3$ via $x$. 
Thus the gadget $G'_2$ is as shown in~\cref{fig:gadget_cases_2-3-4}, containing 
vertices $v_1,v_2,v_3,x$, and the total weight distributed among the 
faces as shown.
The following lemma shows the correctness of this procedure.
\begin{restatable}{lemma}{gadgetCase}\label{lem:gadget_case2}
  Suppose $G_{2}$, which is the subgraph attached to $\Tc$ via $c_2$ has the
  disjoint path configuration as described in Case 2. Suppose it is 
  replaced in $\clos{\Tc}$ by the gadget $G'_2$ defined above for this case. 
  Then if an $\alpha$-\ie-cycle separator $\clos{C}$ of $\clos{\Tc}$ uses any edges of this 
  gadget, we can replace them by paths in $G_{2}$ such that the new cycle
  $C$ is a simple cycle and it remains an $\alpha$-separator in $\Tc \oplus
  G_{2}$. 
\end{restatable}
\begin{proof}
 We defer the proof to~\cref{app:sep_scmfree}
\end{proof}


\textbf{Case 3 : }
\begin{tabular}{c c c}
  $\dpc{v_1}{v_2}{v_3}$ & $\dpc{v_2}{v_3}{v_1}$ & $\dpc{v_3}{v_1}{v_2}$ \\ 
   \fal& \tru & \fal
\end{tabular}

The gadget is shown in~\cref{fig:gadget_cases_2-3-4}.
As in the previous case, we first find the `furthest' vertex $x$ separating $\{v_2,v_3\}$ 
and $\{v_1\}$. 
Then, in 
the component containing 
$v_2,v_3$ (augmented with $x$), 
we find the cut vertex $y$ 
separating $\{x,v_3\}$ and $\{v_2\}$, that is furthest from $v_2$.

\textbf{Case 4 : }
\begin{tabular}{c c c}
  $\dpc{v_1}{v_2}{v_3}$ & $\dpc{v_2}{v_3}{v_1}$ & $\dpc{v_3}{v_1}{v_2}$ \\ 
   \fal & \fal & \fal
\end{tabular}

Like in previous cases, we find the `furthest' cut vertices $x,y,z$ that separate 
$v_1,v_2,v_3$ from rest of graph. There are two cases, either there exists a
single cut vertex $w$ such that $G_{2}-w$ disconnects all of $v_1,v_2,v_3$ from
each other, which is essentially the case when $x=y=z=w$. 
Other case is when there is no such cut vertex. 
The gadget is shown in~\cref{fig:gadget_cases_2-3-4}. The proof of correctness of
gadgets for Cases 3,4 is similar to the proof of~\cref{lem:gadget_case2} presented 
in~\cref{app:sep_scmfree}.
\begin{figure}
\begin{minipage}{\textwidth}
  \includegraphics[scale=0.57]{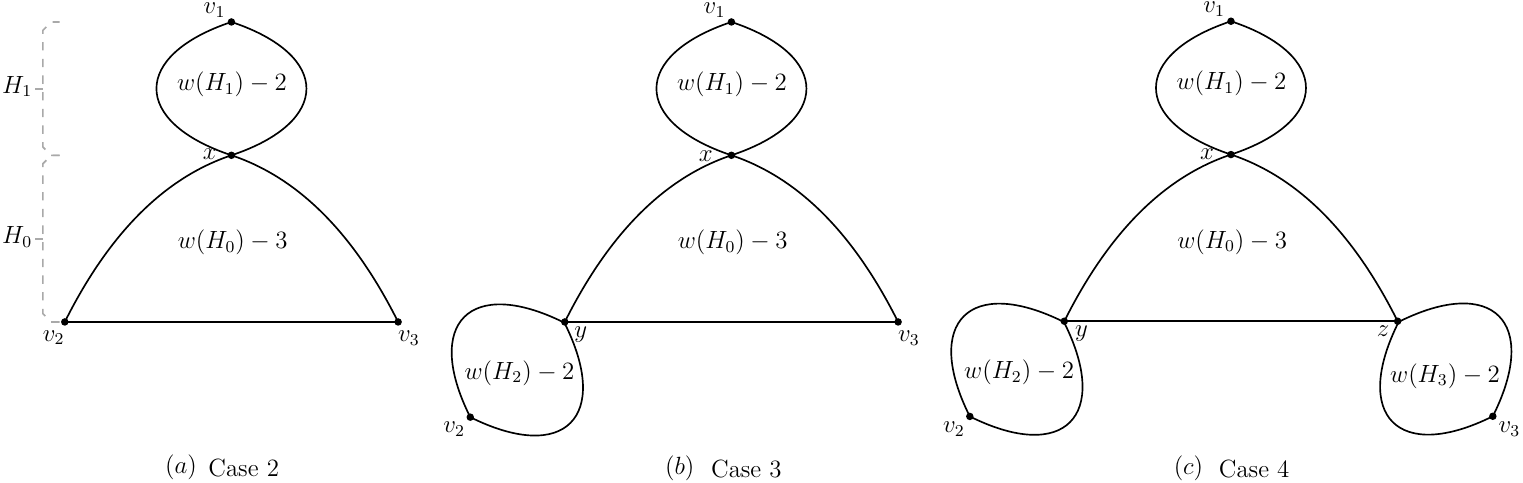}
\caption{Figure (a) shows the gadget for graph $G_{2}$ for Case 2, (b) for
    Case 3 and 
    (c) for Case 4. The faces formed by cycles of the gadget corresponding to 
    the subgraphs $H_0,H_1,H_2,H_3$ of $G_2$ are assigned the weight functions as
    shown. The total weight in each gadgets sums up to $\wt(G_{2})$.
    The subgraph $H_0$ of $G_2$ has no cut vertex in each case. 
    In the gadget for Case 4, it is possible that the block $H_0$ is 
    empty, i.e. there exists a vertex $(x=y=z)$ in $G_2$ which on removal
    disconnects all of $v_1,v_2,v_3$ from each other. We don't draw that gadget
    separately since its structure is clear.}\label{fig:gadget_cases_2-3-4}
\end{minipage}
\end{figure}

To find these cut vertices, we can in parallel check if each vertex
of $G_{2}$ satisfies the properties of vertices `$x,y,z$' described above. 
Thus the gadgets can be constructed in \NC. 

We described the procedure for $G_{2}$, but the construction can be done in
parallel for all subgraphs attached to $\Tc$ via $3$-cliques. 
After finding a $2/3$-\ie-cycle separator of $\clos{\Tc}$ separator using~\cref{thm:IE_cyc_sep},
we can, in parallel, replace the virtual/gadget edges by acutal paths in the attached
subgraphs to get the cycle separator of $G$.

In the blocks of gadgets where we need to find a path between two vertices 
via a third vertex, we can use the $2$-disjoint path
algorithm described in~\cite{KMV92}.
Thus we can construct a $2/3$-path-separator of $G$ in \NC. 


Combined with~\cref{prop:sep_to_DFS}, this finished the proof of part 2
of~\cref{thm:main}. 


\section{DFS in bounded treewidth directed graphs in $\Log$}\label{sec:DFS_btw}

In this section we will show that for a digraph with treewidth bounded by
a constant, 
we can construct a depth first search tree in $\Log$.
For an undirected graph, it is easy to see that a DFS tree of the bidirected 
version of it is also a DFS tree of the original graph.

The input graph $G$ can be seen as a structure
$\struct$ whose universe, $\uod{\struct}$ is $V(G)$, and its
vocabulary
consists of binary relation symbol $E$
which is interpreted as the edge relation $E(G)$ over $V(G)$. The Gaifman graph
of $\struct$, $\gaif{\struct}$, is isomophic to $G$.
We will need to add a linear order $\ord$ on
$\uod{\struct}$, which would amount to adding $|V(G)|$ many edges
in $\gaif{\struct}$. A known result in model theory (see~\cite[Lemma~1]{BD15},
also~\cite[Theorem~\MakeUppercase{\romannumeral 6}.4]{CF12}) is that
we can add a linear order to a structure, such that the treewidth of the structure
is increased by at most a constant.
It is shown in \cite[Lemma~5.2]{Balaji16} that we can construct such a linear order
in $\Log$.
%
Therefore after computing the order $\ord$ and adding it to $\struct$ the new
structure $\struct'$ also has bounded treewidth. 
Moreover, for any two vertices $u,v$, we can query if $u$ is ``lesser than'' $v$
in the transitive closure of $\ord$, in logspace. 

We also add the edges $E(G)$ to our universe of
discourse, as well as an incidence relations : $tail(e,v)$ which is true iff
the edge $e \in E(G)$ has vertex $v$ as its tail, and $head(e,v)$,
which is true iff the edge $e \in E(G)$ has vertex $v$ as its head.
Let the structure obtained by the augmentation of these incidence relations to 
$\struct'$ be denoted by $\struct''$. 
If the treewidth of $\struct'$ is $\tau$, then  
the treewidth of $\struct''$ is at most $\tau+ \tau^2$ (by introducing a new
vertex for each edge, in every bag), and therefore still a constant. 
In order to express an element of $\uod{\struct''}$ being a vertex or
an edge of $G$, we will also need unary relations for the same. These
however clearly do not add to treewidth of $\gaif{\struct''}$, and we
use phrases like $\exists v \in V\ldots , \exists e \in E\ldots$ as
syntactic sugar. Symbols like $u,v$ will generally be used for elements of
universe that correspond to vertices of $G$ and $e$
for elements corresponding to edges of $G$.

  We use $u \lord{\ord} v$ to denote that vertex
$u$ is lesser than vertex $v$ in the transitive closure of $\ord$.
An ordering over vertices defines
a unique lexicographically first DFS tree.
The algorithm will go over each edge of $E$, and use~\cref{thm:EJT_C} to query
if it belongs to the lex-first DFS tree of $G$, in $\Log$. 
Our universe is $V \cup E$, and the relations on
it given are $tail, head$, 
and the linear order $\ord$. To express the membership of an edge in
the lex-first DFS tree with respect to $\ord$, we will use a theorem connecting
the lex-first DFS tree to lex-first paths. 

We define lexicographic ordering on paths starting from a common
vertex:
\begin{definition}
  Let $P_1,P_2$ be two paths in $G$ starting from $r$. We say that $P_1$ is
  lexicographically lesser than $P_2$ (with respect to $\ord$) if at the first point of
  divergence starting from $r$, $P_1$ diverges to a vertex that is lesser in
  the given ordering than the one $P_2$ diverges to. We denote this as
  $P_1 \lord{\ord} P_2$.
\end{definition}
This naturally defines the notion of the lexicographically least path starting
from $r$ and ending at a vertex $v$. We call it the lex-first path from $r$ to $v$.
We use the following result which is shown in the proof of Theorem~11 in~\cite{TK95}:
\begin{theorem}
  Let $\tlex$ be the lex-first DFS tree of $G$ with respect to a given linear
  order $\ord$. Then
  $\forall v \in V(G)$, the unique path from $r$ to $v$ in $\tlex$ is exactly the
  lex-first path from $r$ to $v$ (with respect to $\ord$) in $G$.
\end{theorem}

Thus in order to check if an edge $(u,v)$ belongs to $\tlex$ (rooted at $r$),
it suffices to check if it is the last edge of the lex-min path in $G$ from
$r$ to $v$. We do this by expressing this property in $\mso$ and
using~\cref{thm:EJT_C}.
We now define some formulae (their semantics) to construct the expression of the stated
property.

The variables $P,P_1,P_2$ denote sets of edges:\\ 
$\phiedge{P}{u}{v}$, which is true iff there exists an edge $e \in P$ such that
$e = (u,v)$ and $e \in P$.\\
$\phipath{P}$, which is true iff the edges in $P$ form a single simple
path with $r$ as the starting point.\\ 
$\phiep{P}{v}$, which true iff there exists exactly one in-neighbour of $v$
in $P$, and no out-neighbour of $v$ in $P$.\\ 
$\philex{P_1}{P_2}$, which is true iff $P_1,P_2$ satisfy
$\phipath{P_1},\phipath{P_2}$ respectively, and $P_1\lord{\ord} P_2$.

Using these formulae, we define $\phidfs{u}{v}$, which is true iff
the edge $(u,v)$ is part of the lex-first DFS tree, $\tlex$.
We use the above theorem and express it as:
\begin{equation*}
 \begin{split}
   \phidfs{u}{v} :  \exists P_1 ( & \phiedge{P}{u}{v}
                    \wedge \phipath{P_1} \wedge \phiep{P_1}{v} \wedge \\
   & \forall P_2 ( \phipath{P_2} \wedge \phiep{P_2}{v} \Rightarrow \philex{P_1}{P_2} ) )
 \end{split}
\end{equation*}
Thus the logspace algorithm is the following :\\
For each edge $(u,v)$ in $E$, query $\phidfs{u}{v}$
and add to output tape iff the query returns yes.
A DFS numbering or order of traversal of the vertices of
$\tlex$ can easily be obtained from $\tlex$ by doing an
Euler tree traversal of $\tlex$ with respect to the 
ordering $\ord$.

Now all that remains is to express the formulas $\phiedge{P}{u}{v}, \phipath{P}, \phiep{P}{v}, 
\philex{P_1}{P_2}$ in $\mso$. We defer that to~\cref{app:DFS_btw_app}.

\section{Maximal Paths in planar graphs in $\Log$}\label{sec:maximal_path_undir}
In this section we will show that given a planar undirected graph $G$ and a
vertex $r$ in it, we can find a maximal path in $G$ starting from $r$, in $\Log$. 
We will assume that our graph is embedded in a plane with a face 
marked as outer face (a planar embedding 
can be found in $\Log$ by~\cite{AM04,DP11}). 
The first step is to reduce the problem to that of finding a maximal 
path in a triconnected component of $G$ that is $3$-connected (i.e. it is not a
cycle or a dipole). 
We decompose the graph into its triconnected components, which can be done 
in $\Log$ as described in~\cite{DLNT09}. 
We can choose a leaf piece $\lnode$ of this decomposition tree which is connected 
to its parent piece $\parnode$ by separating pair, say $(r_0,r_1)$.
The basic idea is to find a path from $r$ ($\lnode$ is chosen such that it does 
not contain $r$) to one of the vertices of the separating pair, say $r_0$, 
without going through $r_1$. 
Then it suffices to find a maximal path in $\lnode$ starting from $r_0$, that 
does not end at $r_1$.
Some care must be taken in the reduction however, and cases like when 
$\lnode$ is a simple cycle need to be handled.  
\subsection{Reduction from planar graphs to its triconnected components}

\paragraph*{Reduction to biconnected graphs.}
In the first step of the reduction, we reduce the problem to the case of biconnected graphs.
(this step is also done in~\cite{AM87}). 
This can be done easily by decomposing the original graph into its biconnected components
using its block-cut tree. Root the tree at the
block that contains $r$, and consider any block $H$ corresponding
to a leaf node of the block cut tree. Let $r'$ be the articulation vertex of $H$.
Then finding a path in the graph from $r$ to $r'$, and then appending to it
a maximal path of $H$ starting from $r'$, will clearly give us a maximal
path of the input graph. The reduction is in $\Log$ using Reingold's 
algorithm~\cite{Reingold08} 
along with (constantly many) transducers. Henceforth, we assume that our input 
graph $G$ is a biconnected graph.

\paragraph*{Reduction to tri-connected components.} 

Given that our input graph $G$ is biconnected, we now decompose $G$ further 
by $2$-clique sums into triconnected components,
using the $\logs$ algorithm described in~\cite{DLNTW22}. 
When looking the embedding of a piece by itself in isolation (that is ignoring the rest
of the graph), it's outer face is naturally inherited from the outer face 
face of $G$. 
In particular, the separating pairs of a piece and the edges connecting them
(possibly virtual edges) all lie on the outer face in its own standalone
embedding.  Root the decomposition
tree $T_{G}$, at the piece containing $r$. We will pick a leaf piece 
which is not a triple bond (It can be seen in the decomposition
of~\cite{DLNTW22} that not every leaf node of the triconnected decomposition
tree can be a triple bond. In fact every triple bond piece is accompanied with
a sibling piece that is either a cycle or a $3$-connected graph. Hence we can ignore
the triple bond pieces for our purposes) and reduce the problem to finding a maximal path
in that, possibly with some additional constraints.  

Let $L$ denote the leaf piece that we will 
pick to reduce the problem to. We will assume that $L$ has parent piece
$\parnode$, and $L$ is attached to $\parnode$ via separating pair $c=\{r_0,r_1\}$. 
In the piece $\parnode$, the separating 
pair $c$ lies on the boundary of its outer face 
$f_{\parnode}$ (see~\cref{fig:maximal_leaf}). 
We will use $G'$ to denote the graph obtained by augmenting $G-L$ with $c$ (that
is removing $L$ from $G$ but keeping vertices, edges of $c$ intact). Note that
$G'$ will be biconnected (see~\cref{obs:outer_cycle}).
We consider the following cases.
\begin{figure}[t]
\begin{minipage}{\textwidth}
  \includegraphics[scale=0.55]{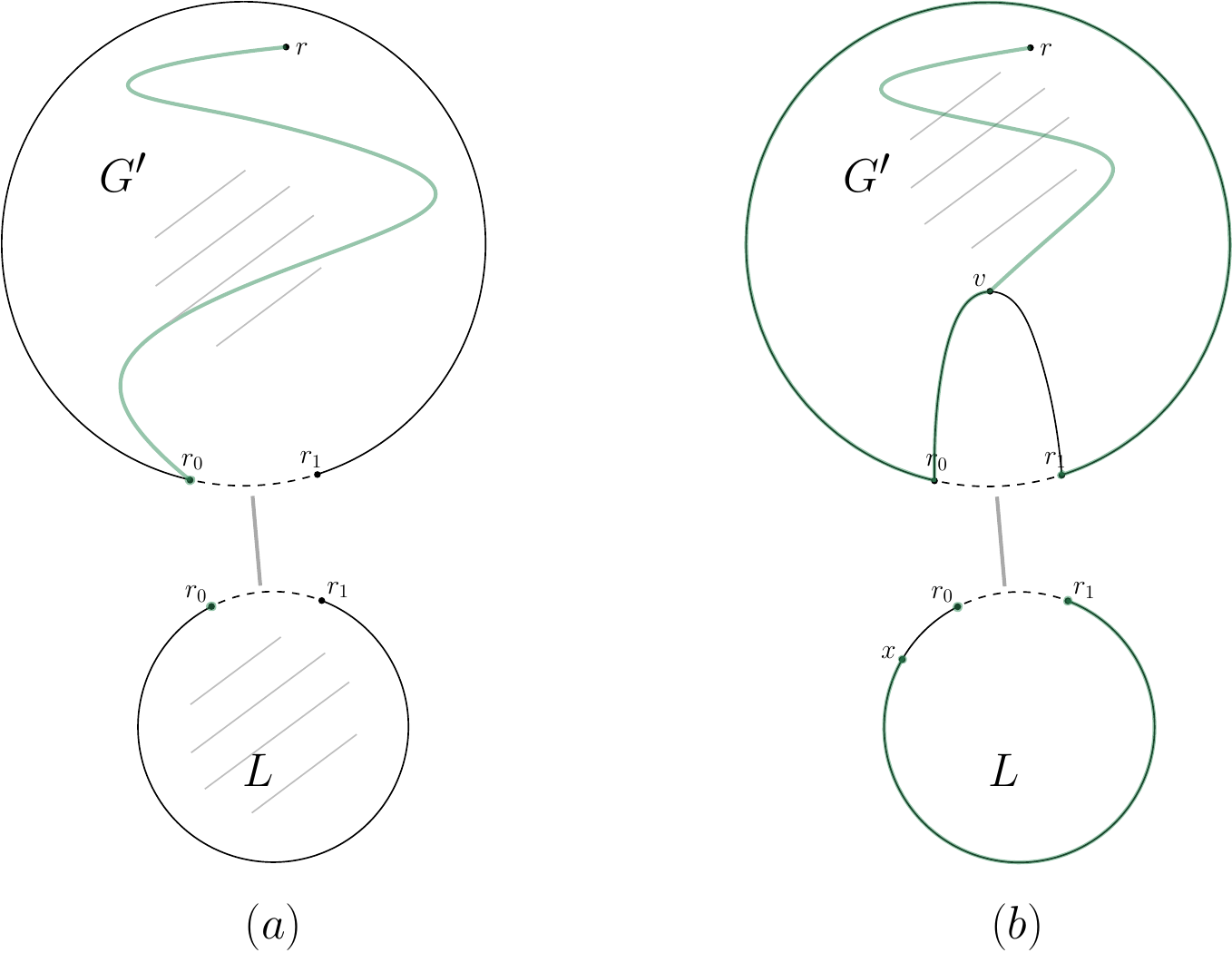}
  \caption{ Both cases for reduction of maximal path problem to the triconnected
  components of $G$. Dashed edges denote virtual edges in $\Tc$. 
  The edge $(r_0,r_1)$ may be real or virtual. In these figures it is virtual.
  Figure (a) shows 
  the case when there is a leaf node $L$ that
  is a $3$-connected graph. 
   From $r$ we find a path in $G'$ to reach 
   $r_0$ as shown in figure (a). This reduces the problem to
  finding a maximal path in $L$ that does not end at $r_1$ and does not use the
  virtual edge $(r_0,r_1)$.\\ 
  Figure (b) illustrates the case when all the leaf nodes (and hence $L$) are cycles. 
We find a path from $r$ touching the boundary of $G'$ (which must be a simple
cycle) for the first time at $v$. From $v$ we encircle the boundary of $G'$
first reaching $r_0$ and finally $r_1$. From $r_1$ we encircle the cycle $L$
till $x$. Since both the neighbours of $x$ are visited, this path is maximal.
}\label{fig:maximal_leaf}
\end{minipage}
\end{figure}

\begin{itemize}
  
  \item \emph{Case 1: There exists a leaf piece in $T_G$ that is a non-trivial $3$-connected
    graph (that is, not a cycle).} 

    In that case we pick such a piece as $L$, and look at the graph $G'$ which
    contains $r$ (see part $(a)$ in \cref{fig:maximal_leaf}).  In $G'$, find a
    path $P_0$ from $r$ to $r_0$ that does not go through $r_1$.
    We can find such a path by finding a path from $r$ to $r_0$ in $G'-r_1$. 
    Since $G'$ is biconnected, such a path will exist. 
    Now its suffices to find a maximal path in $L$ starting
    from $r_0$, that does not end at $r_1$ and does not use the virtual edge
    $(r_0,r_1)$, since we can append it to $P_0$ and 
    get a maximal path in $G$. Thus the reduction is complete in this case. 
  
  \item Case 2: \emph{None of the leaf pieces of $T_G$ is a non-trivial
    $3$-connected graph.} 

   This means they are all cycles (we ignore dipoles). Consider any leaf
   piece $L$ that is a cycle and is attached to the parent $\parnode$ at $r_0,
   r_1$ (see part $(b)$ of \cref{fig:maximal_leaf}), 
   Since $G'$ is biconnected, the boundary of its outer face is a simple cycle.
   Therefor $r_0,r_1$ have two disjoint segments connecting them, both part of
   the outer boundary. From $r$, we can find a path $P_0$ which touches the
   boundary for the first time at its other end point, say $v$. From $v$, we 
   append the path $P_1$ by circling around the boundary until we first reach, say
   $r_0$, and then $r_1$. From $r_1$, we can find a path $P_2$ in $L$, by
   encircling it till the neighbour of $r_0$ other than $r_1$ ($L$ is a simple
   cycle). Clearly $P_0.P_1.P_2$ is a maximal path and each of these can be
   constructed in \Log. Hence we are done in this case. 
\end{itemize}
Each of the steps above, computing the decomposition 
into triconnected pieces, constructing paths $P_0,P_1,P_2$ 
can be done in in $\Log$ using the algorithms of~\cite{Reingold08,DLNTW22}.
The output of the reduction steps of Case 1 (The path $P_0$, the piece $L$, $r_0,r_1$) 
can be stored on a transducer tape and we can proceed further.

\subsection{Maximal Paths in $3$-connected components}
By the reduction above, we can assume that our input graph is a 
$3$-connected component of $G$, say $L$, and 
our goal is to find a maximal path in $\lnode$ 
starting at $r_0$ and not ending at $r_1$. 
Note that the latter condition makes this problem more general than just 
finding a maximal path in a $3$-connected graph. 
Also, note that $r_0,r_1$ lie on a common face of $L$, and the edge 
$(r_0,r_1)$ might be a virtual edge in $\lnode$.
Hence we will in general look at the graph 
$L' = L - (r_0,r_1)$, that is, the graph obtained from $L$ by removing the virtual 
edge $(r_0,r_1)$ (not the vertices). 

Let the outer face of $L$, on which both $r_0,r_1$ lie, be 
called $\oface$. The outer face of $L'$ is naturally obtained by merging 
$\oface$ with the other face adjacent to $(r_0,r_1$). 
Let this 
be called $\ofacex$. 
We first note the following observation: 
\begin{restatable}{observation}{outerCycle}\label{obs:outer_cycle}
 The boundary of the face $\ofacex$ of $\lnode'$ is a simple cycle.
\end{restatable}
If not, there would be a cut vertex $x$ in $\lnode'$. 
Then $(r_0, x)$ would form a 
separating pair in $\lnode$, contradicting the fact that it is $3$-connected.
We hereby denote this cycle, that is the boundary of $L'$, by $C$. 
We recall the definition of bridges of a cycle in a graph, which
(roughly) are components of the graph minus the cycle, along with their
attachments to the cycle.
\begin{definition}{\cite{tutte}}
 For a subgraph H of G, a vertex of attachment of H is a vertex of H 
 that is incident with some edge of G not belonging to H. Let J be an 
 undirected cycle of G. We define a bridge of J in G as a subgraph B of G 
 with the following properties:
 \begin{enumerate}
  \item each vertex of attachment of B is a vertex of J .
  \item B is not a subgraph of J.
  \item no proper subgraph of B has both the above properties.
 \end{enumerate}
\end{definition}
A bridge could also be a chord of $C$. We call such bridges as \emph{trivial} bridges, 
and others as \emph{non-trivial} bridges.
We need the following observation for $L'$ and $C$:
\begin{observation}
  Any non-trivial bridge $B$ of $C$ (in the graph $L'$), if one exists, 
  must have at least three vertices of attachment.
\end{observation}
  This holds because if a non-trivial bridge of $C$ in $L'$ had only two vertices
  of attachment on $C$, they would form a separating pair of $L$. 
There is a natural ordering we can give to the vertices of attachments
of a bridge, that is the order in which they are encountered while 
travelling $C$, say clockwise, starting from $r_0$. 
We now define the $\emph{span}$ of a bridge.
\begin{definition}\label{def:bridge_span}
  Let the vertices of cycle $C$ be $\langle u_0=r_0,u_1,u_2 \ldots u_l \rangle$ in
  clockwise order in the given embedding.
  Let $B$ be a bridge of $C$ with vertices of attachment $u_{i_0},u_{i_1},\ldots 
  u_{i_k}$, where 
  $0 \leq i_0 < i_1 \ldots i_k \leq l$.
  Then the set of vertices $\{u_{i_0},u_{i_0+1},\ldots u_{i_k}\}$ 
  is called the \emph{span} of bridge $B$ with respect to $r_0$ (or just \emph{span 
  of $B$} for brevity), denoted by $\lspan{B}$. 
  The set of vertices $\{u_{i_k},u_{i_k+1}\ldots u_l, u_0,\ldots u_{i_0}\}$ 
  is called the \emph{complement span} of bridge $B$ (with respect to $r_0$). 
  We say that bridge $B$ has minimal span if for any other bridge $B'$ of $C$,
  $\lspan{B'} \nsubseteq \lspan{B}$.  
\end{definition}
Basically, the vertices of attachment of a bridge divide $C$ into segments. 
The span of a bridge consists of vertices of $C$, exluding the ones in the
segment containing $r_0$, while the complement span consists of the segment
containing $r_0$. To disambiguate the case when $r_0$ itself is a vertex of
attachment of a bridge $B$, note that the definition uses the convention of
putting the segment adjacent to $r_0$ on the clockwise side into the span of
$B$, and the segment adjacent to $r_0$ on the counter clockwise side into the 
complement span of $B$.
Since $G$ is planar and $C$ is outer boundary, 
the spans of bridges form a laminar family. 
\begin{observation}\label{obs:span_laminar}
  Let $B_1,B_2$ be two bridges of $C$, lying in the interior of $C$. Then either
 $\lspan{B_1} \subset \lspan{B_2}$, or 
 $\lspan{B_2} \subset \lspan{B_1}$, or
 $\lspan{B_1} \cap\lspan{B_2} = \emptyset$.
\end{observation}
Now we define the \emph{second-leftmost} path of a bridge of $C$.
\begin{definition}\label{def:sleftm_path}
  Let $B$ be a non-trivial bridge of cycle $C$, with vertices of attachment
  $u_{i_0},u_{i_1}, u_{i_2}\ldots u_{i_k}$ according to the ordering specified
  above. 
  The $\emph{second leftmost}$ path of $B$ starting from $u_{i_0}$, is
   the one obtained by the following procedure:
   \begin{itemize}
    \item Initialize $u_{i_0}$ as current vertex and $u_{i_0+1}$ as 
      its $\emph{parent}$.
    \item At any step, let the current vertex be $v$, and its parent be $w$.
      Let $z$ be the next vertex occuring after $w$ \emph{other than the vertex
      $u_{i_1}$}, in the clockwise 
      ordering of neighbours of $v$ in the embedding. 
      Travel to $z$ and make it the current vertex. Make $v$ its parent.
    \item Repeat the above step, until either a vertex repeats in the traversal, or a
      vertex of $C$ (after the initial occurence of $u_{i_0}$) is reached.
   \end{itemize}
   We denote this path by $\sleftm{B}{u_{i_0}}$.
\end{definition}
\begin{figure}\label{fig:maximal_tricon}
\begin{minipage}{\textwidth}
  \includegraphics[scale=0.7]{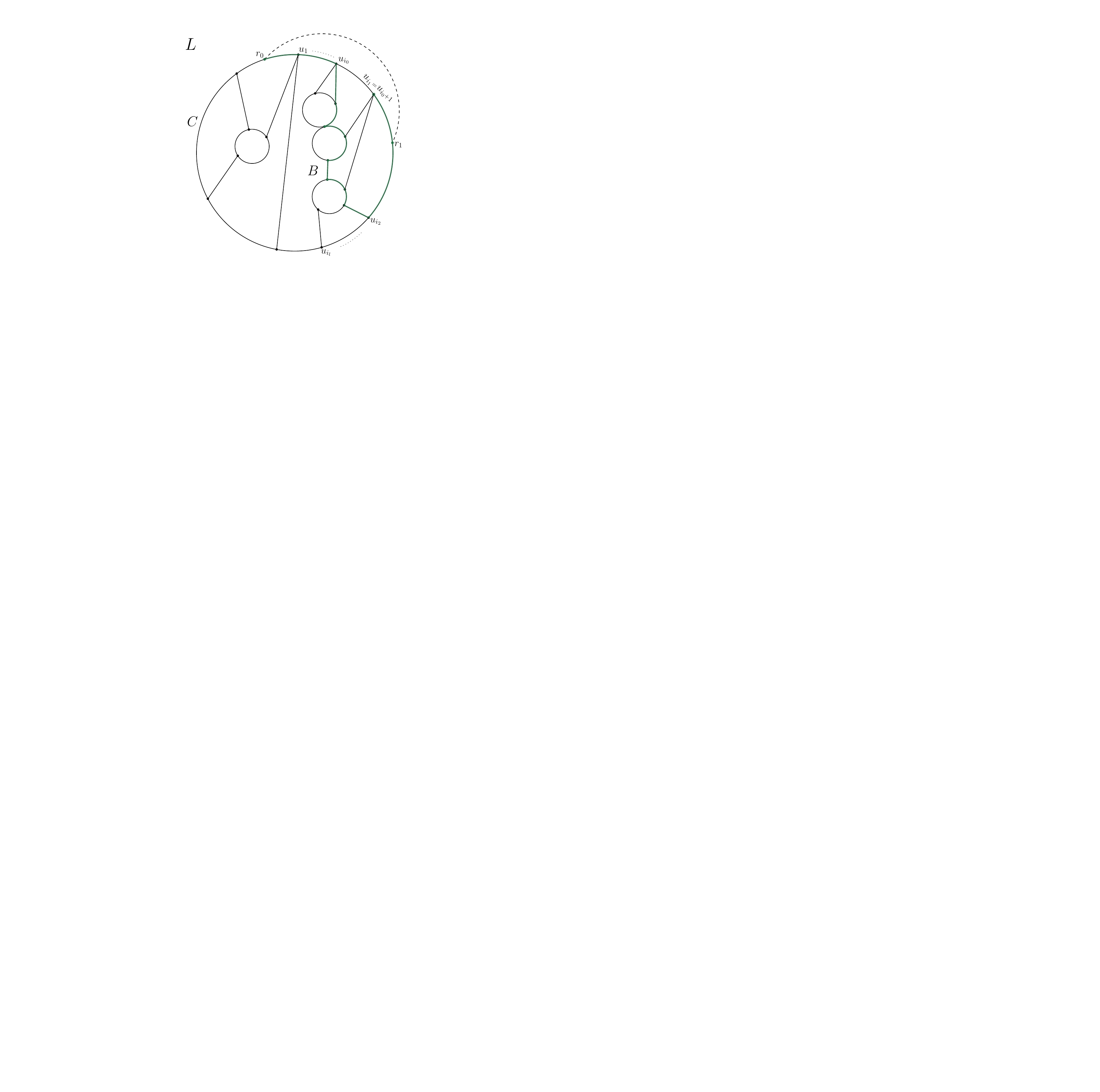}
  \caption{The $3$-connected graph $L$. The dashed edge between $r_0,r_1$ is 
    a virtual edge. The boundary of $L-(r_0,r_1)$ is the cycle $C$. 
    The bridge on the right side is the labelled as $B$, with vertices of attachment 
    $u_{i_0}, u_{i_1}, u_{i_2}, \ldots u_{i_l}$. The maximal path is highlighted 
    in green, starting from $r_0$, going till $u_{i_0}$, then taking the second 
    leftmost path in $B$ to reach $u_{i_2}$, and then encircling back towards 
    $u_{i_0}$, stopping at its neighbour $u_{i_0+1}$ (In this figure, $u_{i_1}$ 
    is same as $u_{i_0+1}$).}\label{fig:maximal_tricon}
\end{minipage}
\end{figure}

We next state a lemma that the second-leftmost path of a non-trivial bridge, 
starting at $u_{i_0}$,
is indeed a simple path that ends at $u_{i_2}$ because $L$ is $3$-connected.
\begin{restatable}{lemma}{sleftPath}\label{lem:sleft_path}
  Given a non-trivial bridge $B$ of cycle $C$, with vertices of attachment $u_{i_0},u_{i_1},
  u_{i_2}\ldots u_{i_k}$ as described above, the following hold:
  \begin{enumerate}
    \item $\sleftm{B}{u_{i_0}}$ is a simple path that ends at $u_{i_2}$.
    \item All neighbours of $u_{i_1}$ that lie in $B$, also lie on
      $\sleftm{B}{u_{i_0}}$.
  \end{enumerate}
\end{restatable}
\begin{proof}
   \begin{enumerate}
   \item We first show that a vertex cannot repeat in walk described for $\sleftm{B}{u_{i_0}}$. 
    Suppose $v_a$ is the first vertex that repeats. We can then write the walk
    till the first repetition as
    $\sleftm{B}{u_{i_0}} = u_{i_0},v_1,v_2, \ldots v_{a-1}, v_a,v_{a+1} \ldots v_b, v_a$.
    Consider the cycle $C' = v_a,v_{a+1} \ldots v_b, v_a$ formed by the walk.
    It must be a non-trivial cycle, for if it were just an edge, with $v_{a+2}=v_a$,
    then $v_{a+1}$ would be a pendant vertex.
      Since $C'$ is contained in the interior of $C$, and it is traversed starting
    from the parent edge $(v_{a-1},v_a)$ which lies in the exterior of $C'$,
    the only edges other than $(v_{a-1},v_a)$ incident to $C'$, and lying in
    its exterior are (possibly) the ones with $u_{i_1}$ as the other end point 
    (see~\cref{fig:maximal_sleftm}).
    If any edge other than these were incident to $C'$ from its exterior, 
    then we would have traversed
    that edge instead of $C'$ by our definition of $\sleftm{B}{u_{i_0}}$.
    This implies that either $v_a$ is a cut vertex for $C'$, or $(v_a,u_{i_1})$
    is a separating pair disconnecting $C'$ from $C$, violating that $L$ is $3$-connected.
    Therefore the walk described for $\sleftm{B}{u_{i_0}}$ must end at a vertex
    of $C$. Next we show that the vertex it ends at is $u_{i_2}$. 
    It can clearly not end at $u_{i_1}$ by definition of the path.
    Therefore it must end at one of $u_{i_2},u_{i_3},\ldots u_{i_k}$. Suppose
    it ends at $u_{i_x}$, where $i_x>i_2$. Then $\sleftm{B}{u_{i_0}}$ would form a
    closed loop along with the segment $\segr{u_{i_0}}{u_{i_x}}$ of $C$. Since
    $u_{i_2}$ is a vertex of attachment of $B$, there must be a path 
    from $u_{i_2}$ to $\sleftm{B}{u_{i_0}}$ in the interior of this loop which
    again is not possible by definition of the path, as we always took the 
    clockwise first choice as long as it doesn't arrive at $u_{i_1}$.  
    Therefore the end point of $\sleftm{B}{u_{i_0}}$ in $G'$ 
    must be $u_{i_2}$.
   \begin{figure}
\begin{minipage}{\textwidth}
  \includegraphics[scale=0.7]{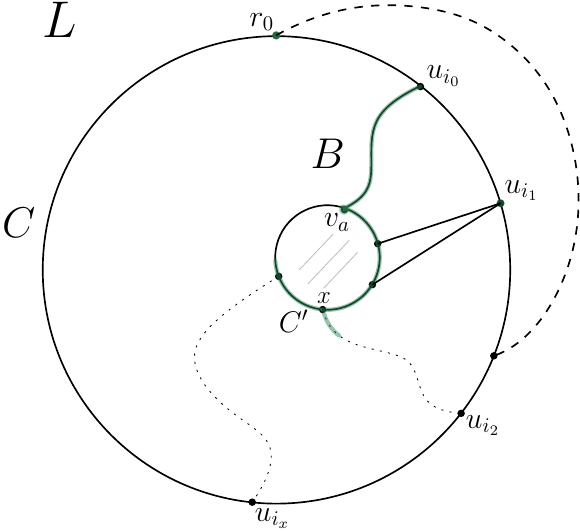}
  \caption{ The graph $L$ with $C$ as the outer boundary of $L'$ (the dashed
    edge outside the outer face is the virtual edge $(r_0,r_1)$). The second leftmost
    walk $\leftm{B}{u_{i_0}}$ is highlighted in green. If the vertex $v_a$ were
    to be repeated in the walk, then edges shown in dotted, like the one
    incident on $x$ from exterior of $C'$ cannot be present, else the walk would
    have branched away from $C'$ at $x$. In that case, $(v_a,u_{i_1})$ clearly
    form a separating pair of $L$.
}\label{fig:maximal_sleftm}
\end{minipage}
\end{figure}
  \item This again follows from the fact that $\sleftm{B}{u_{i_0}}$ forms a closed
    loop with the segment $\segr{u_{i_0}}{u_{i_2}}$. Since $u_{i_1}$ lies on the
    outer boundary $C$of $G'$, its neighbours must lie on or inside this closed
    loop. The vertices that are common between this closed loop and $B$ are
    exactly those of $\sleftm{B}{u_{i_0}}$, and no vertex of $B$ can lie in the
    interior of this closed loop by the above proofs. Therefore the claim holds. 
\end{enumerate}

\end{proof}
Since we know that $\sleftm{B}{u_{i_0}}$ is a simple path starting and ending at $C$,
it is clear that we can perform the traversal to construct it in $\L$, 
by just storing the
current vertex, the parent vertex, and the vertex $u_{i_1}$ at each step.
The reason we choose the \emph{second leftmost path} in the bridge is so that 
after branching from the cycle $C$ into its interior at $u_{i_0}$, we can 
emerge back at $C$ at $u_{i_2}$, and then 
trap the maximal path in the segment $\segl{u_{i_2}}{u_{i_0}}$ of $C$. 
The existence of the unvisited vertex $u_{i_1}$ ensures that the segment 
is non-empty and we can do so.
In the case when the bridge $B$ is a chord, we define its second leftmost path 
as just the chord itself.

We now describe the algorithm to compute maximal path in $L$ starting 
from $r_0$, with a small exception that we will handle subsequently.
\begin{enumerate}
  \item Compute the bridges of $C$ (w.r.t. $L'$) that lie in the
    interior of $C$ by computing the components of $G-C$ along with their 
    vertices of attachment to $C$. 
  \item From $r_0=u_0$, extend the path by going clockwise till we find the 
    first point of attachment 
    of a bridge $B$ of $C$ with minimal span. 
    Let the vertices of attachment of $B$, in clockwise order starting from $r_0$
    be $u_{i_0},u_{i_1},u_{i_2}, \ldots u_{i_k}$ (where ${i_k} \geq 2$).
  \item From $u_{i_0}$, continue along the path $\sleftm{B}{u_{i_0}}$ till 
    we reach $u_{i_2}$ (we will reach $u_{i_2}$ by~\cref{lem:sleft_path}). 
    From $u_{i_2}$, extend this
    path by traversing the segment 
    $\segl{u_{i_2}}{u_{i_0+1}}$, 
    till the vertex $u_{i_0+1}$. Suppose that $u_{i_0+1} \neq r_1$. Then this 
    path is a maximal path.
\end{enumerate}

To see that this path is maximal (assuming $u_{i_0+1} \neq r_1$), 
note that no other bridge of $C$ can be 
enclosed inside $\sleftm{B}{u_{i_0}}$ and the segment 
$\segl{u_{i_2}}{u_{i_0}}$ of $C$,
since we chose $B$ with a minimal span. 
There are two possible cases, either $u_{i_0+1}$ has no neighbour in the 
interior of $C$ in which case we are done. 
The only other possiblity, by the above argument is that $B$ is a non-trivial 
bridge and $u_{i_0+1}=u_{i_1}$ ($B$ cannot be a chord between $u_{i_0}$ and
$u_{i_0+1}$ as we assume input graph to be simple).
In this case, by~\cref{lem:sleft_path}, its neighbours inside $C$ are all visited in 
$\sleftm{B}{u_{i_0}}$, and we are done.
Now we handle the other case.

\paragraph*{Case : $u_{i_0+1}=r_1$}
To handle this case, we note that in the algorithm we saw for previous case, 
we only need \emph{some} bridge with minimal span, not necessarily the first one we
encounter while travelling clockwise on $C$ from $r_0$. So the first step we do
is to check if there exists another bridge $B'$ with minimal span (hence
disjoint from the span of $B$). 
\begin{itemize} 
  \item \textbf{Subcase 2.1} : There exists another bridge $B'$ with minimal span.\\
    We travel from $r_0$ clockwise to the first
vertex of attachment of $B'$ and repeat the same algorithm as in previous case. 
It is easy to see that the new ending vertex we obtain by this tweak 
cannot be the same as $r_1$, since the spans of $B,B'$ 
are disjoint and $r_1$ lies strictly inside the span of $B$.

  \item \textbf{Subcase 2.2} : There exists no bridge other than $B$ with minimal span.\\
    By~\cref{obs:span_laminar}, this implies that the spans of all bridges of $C$ form  
a linear order in the subset relation. Then the complement spans of all these
bridges must also have a linear order. In this case we travel clockwise along $C$ 
starting from $r_0$, until we reach the first vertex $u'_{i_0}$ which is a
vertex of attachment of the bridge $B'$ with minimum complement span. From
$u'_{i_0}$, with $u'_{i_0-1}$ as the parent (or $u_l$ as the parent in case
$u'_{i_0}=r_0$), we take the \emph{second rightmost
path} (definied similarly as the second leftmost path with appropriate tweaks)
to reach a vertex $u'_{i_2}$ on $C$ (see~\cref{fig:maximal_tricon_case2}).
Since $u_{i_0}$ is the first vertex of attachment of $B'$ clockwise from 
$r_0$, $u'_{i_2}$ cannot be the same as $u_l$. From $u'_{i_2}$, we traverse 
the segment $\segr{u'_{i_2}}{u_l}$ and finish at vertex $u_l$. 
Note that $u_l$ lies in the complement span of $B'$.
The corresponding statement of~\cref{lem:sleft_path} for second rightmost
path holds with a similar proof.
The argument that the path constructed here is a maximal path therefore is also 
similar to the one presented in~\cref{sec:maximal_path_undir} for the previous
case. 
\begin{figure}[t]
\begin{minipage}{\textwidth}
  \includegraphics[scale=0.7]{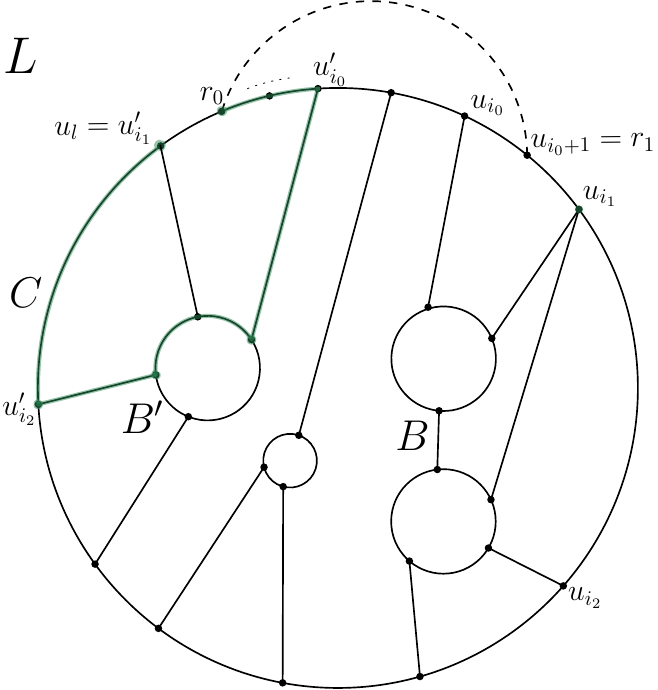}
  \caption{The $3$-connected graph $L$. The dashed edge between $r_0,r_1$ is 
    a virtual edge. The boundary of $L-(r_0,r_1)$ is the cycle $C$. 
    If we take the second leftmost path $\sleftm{B}{u_{i_0}}$ along the bridge $B$, 
    we end at the vertex $u_{i_0+1}$ which is the same as $r_1$, which is
    forbidden for us. We therefore consider the bridge $B'$ which has the
    minimum \emph{complement span}. From its first vertex of attachment
    $u'_{i_0}$ clockwise from $r_0$, we take the second leftmost path
    (highlighted in green) in $B'$ to reach $u'_{i_2}$, and then encircle back
    clockwise towards $r_{0}$, stopping at its neighbour $u_{l}$. 
    In this figure, $u_{1}$ also happens to be a vertex of attachment of $B'$
    (hence marked same as $u'_{i_1}$). Its neighbours in $B$ are already visited
    in the second rightmost path as seen in the figure.}\label{fig:maximal_tricon_case2}
\end{minipage}
\end{figure}
Since $u_l$ lies in complement span of $B'$, and $r_1$ lies strictly inside the
span of $B$, $u_l$ must be disctinct from $r_1$. 
\end{itemize}
We can sequentially go through the cases in $\Log$ using constant number of transducers. 
This finishes the proof.

Each of the steps in the algorithm like computing bridges, checking for minimal
span, can be done in $\Log$ using constantly many transducers
and~\cite{Reingold08}.
And as explained above, computing the required segments of $C$, and the 
second leftmost path, can also be done in $\Log$. 
Hence the entire algorithm is in $\Log$.

\bibliography{references}
\appendix
\section{Appendix}
\subsection{Appendix for~\cref{sec:path_separators}}\label{app_path_separators_app}
\subsubsection{Proof of~\cref{thm:Kao}}\label{app:thm_Kao_proof} 
We restate the theorem here:
\thmKao*
\begin{proof}
  Let the two paths that together form an $\alpha$ separator be
  $P_1= \{u_1,u_2 \dots u_p\}$ and $P_2 = \{v_1,v_2\ldots v_q\}$.
  \begin{itemize}
    \item Trim the path $P_1$ to $P'_1 = \{u_1,u_2 \dots u_s\}$ $(s \leq p)$
      such that separator property is just lost on trimming away $u_{s-1}$. That is, $P'_1
      \cup P_2$ is an $\alpha$-separator of $G$, but $P'_1\backslash \{u_s\} \cup P_2$
      is not. This means that in $G - (P'_1\backslash \{u_s\} \cup P_2)$,
      the vertex $u_{s}$ is
      a part of a strongly connected component of size at least $\alpha n$.
    \item Having trimmed $P_1$ to $P'_1$, repeat the same step and trim $P_2$ to
      $P'_2 = \{v_t,v_{t+1} \dots v_q\}$ ($t\geq 1$) such that
      $P'_1 \cup P'_2$ forms an $\alpha$-separator, and the vertex
      $v_t$ is part of some strongly connected
      component of size at least $\alpha n$
      in the graph $G - (P'_1 \cup P'_2\backslash \{v_t\})$.
    \item Thus in $G - (P'_1\backslash \{u_s\} \cup P'_2\backslash \{v_t\})$, there
      are two strongly connected components of size at least
      $\alpha n$ ($\alpha \geq 1/2$), containing $u_s, v_t$ resepectively.
      This means that they must be part of a single strongly connected component
      and hence there must be a path $P$ from $u_s$ to $v_t$ in
      $G - (P'_1\backslash \{u_s\} \cup P'_2\backslash \{v_t\})$. Therefore we can
      concatenate the paths $P'_1, P, P'_2$ into a single path $P'_1.P.P'_2$,
      which is an $\alpha$-separator of $G$.
  \end{itemize}
  While trimming paths in steps $1$ and $2$, we only need to remember the current end point
  of the trimmed path, which can be done in $\Log$.
  Other steps, like checking the size of leftover connected components
  and
  finding a path from $u_s$ to $v_t$ in the subgraph can be done in $\Log$ given
  an oracle to $\reach$.
\end{proof}  

\subsection{Appendix for~\cref{sec:sep_scmfree}}\label{app:sep_scmfree}

We demonstrate the argument for correctness of the gadget in Case 2 here. 
The argument is similar for the other 2 cases. We restate the lemma: 

\gadgetCase*

\begin{figure}
\begin{minipage}{\textwidth}
  \includegraphics[scale=0.7]{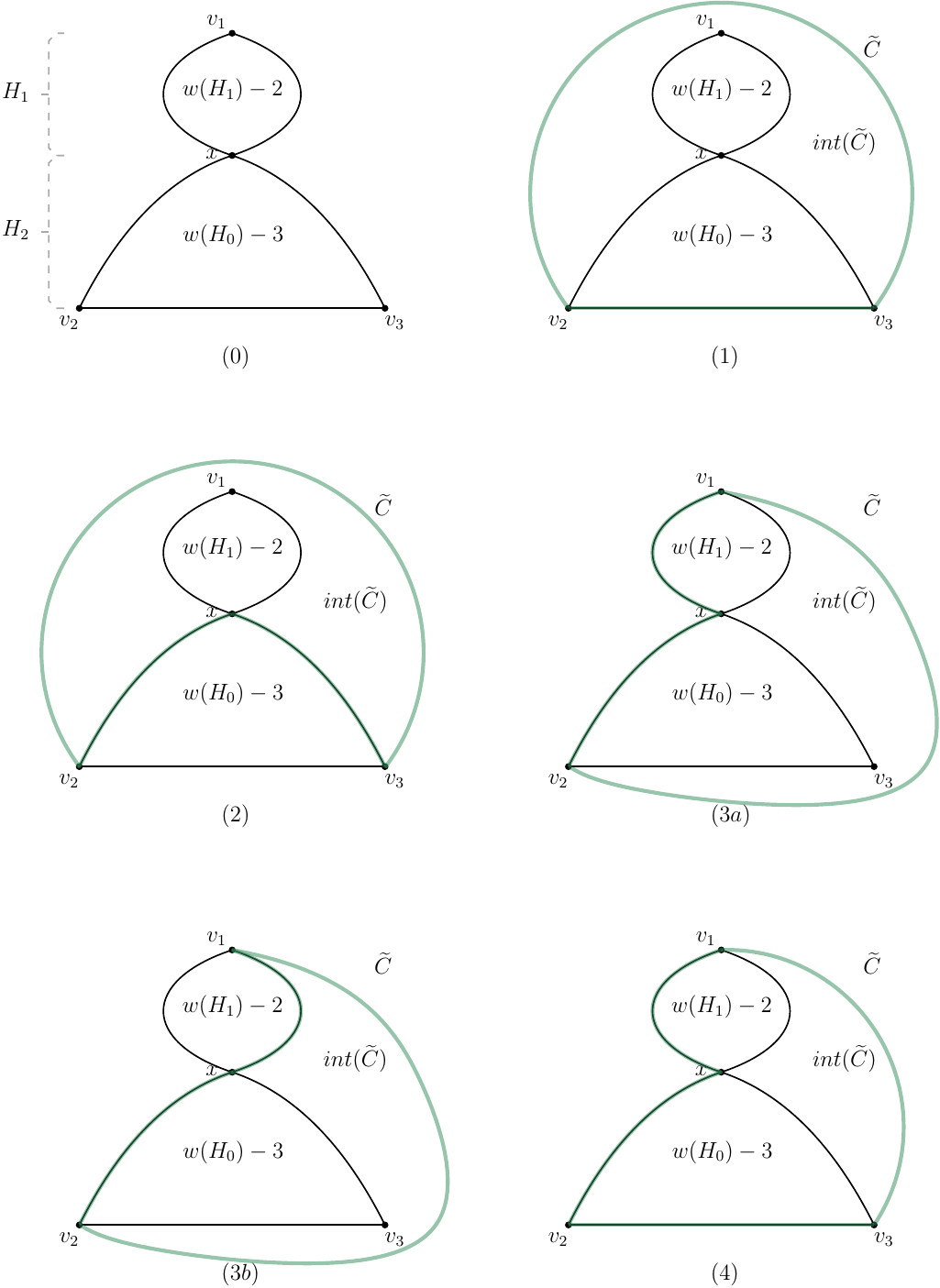}
  \caption{Figure (0) shows the gadget for the graph $G_{2}$ for Case 2. 
    The values $2,3$ are subtracted from weight of the faces because the 
    vertices $v_1,v_2,v_3,x$ have a weight of one themselves. 
    Note that $\wt(H_1)+\wt(H_2) = \wt(G_{2})+1$ (since $x$ is shared by 
    $H_1$ as well as $H_0$). 
    The subgraph $H_0$ of $G_2$ has no cut vertices. 
    The remaining figures show 
    the possible ways the cycle $\clos{C}$ can use the segments of the 
    gadget. They correspond to the subcases $1-4$ of the argument to show the 
    correctness of the gadget in Case 2.}\label{fig:gadget_proof}
\end{minipage}
\end{figure}

\begin{proof}
Suppose $\clos{C}$ uses some segment of this gadget. 
Without loss of generality, we can assume that each of $v_1,v_2,v_3$ lies on 
$\clos{C}$, or in the interior of $\clos{C}$. Therefore $G_{2}$ is an
interior component of $\Tc$ with respect to $\clos{C}$.
We need to argue two things. First, that for the segments of the gadget used by
$\clos{C}$, we can substitute simple paths in $G_{2}$.
Second, that after the substitution, the new cycle remains a balanced separator. 
For this, it is sufficient to show that the weight contribution of $G_{2}$ to
the interior components is at most the weight contribution of the gadget to
$\wt(int(\clos{C}))$, and the weight contribution of $G_{2}$ to the exterior
components is at most the weight contribution of the gadget to 
$\wt(ext(\clos{C}))$ (we will drop the term `weight' and just use contribution
of $G_{2}$/gadget hereby for brevity).
This is because by the property of cycle separator, both 
$\wt(int(\clos{C})), \wt(ext(\clos{C}))$ are already guaranteed to be at most
$2n/3$, and as explained above, there can be no path from any interior component
of $\Tc$ in $G-V(\clos{C})$ to any exterior component of $\Tc$ in
$G-V(\clos{C})$.

Since $G_{2}$ is an interior component, its contribution to exterior components
of $\Tc$ is $0$ regardless of which segment of the gadget is taken by $\clos{\Tc}$. 
Therefore we only have to argue about its contribution to
interior components. 

Also note that in any case, if all three of $v_1,v_2,v_3$ lie on the cycle
$\clos{C}$, then also the contribution of $G_{2}$ is $0$ to both interior as
well as exterior components of $\Tc$, as $G_{2}$ gets disconnected from 
rest of the graph on removing vertices of $\clos{C}$. Therefore for the second
part of the proof about the cycle remaining a balanced separator, we only have
to consider subcases when one of $v_1,v_2,v_3$ lies in the strict interior of
$\clos{C}$.

Now we look at the subcases of which segment of the gadget is used by $\clos{C}$.
There are four possible subcases 
(up to symmetry):
\begin{enumerate}
  \item $\clos{C}$ uses only the edge $(v_2,v_3)$ from the gadget 
    (see subfigure $(1)$ in~\cref{fig:gadget_proof}).
    Replace the edge 
    $(v_2,v_3)$ with any $v_2\text{-} v_3$ path in $G_{2}$ in the new cycle 
    $C$. It will clearly remain a simple cycle. 
    By the explanation above, we can assume that $v_1$ lies in 
    the \emph{strict} interior of $\clos{C}$.
  Then the contribution of the gadget to $\wt(\clos{C})$ is $\wt(G_{2})-2$,  
    which is at least as much as contribution of $G_{2}$ to interior components 
    of $\Tc$. 
    Therefore if $\clos{C}$ is an $\alpha$ \ie-cycle separator of $\clos{\Tc}$, then 
    $C$ is certainly an $\alpha$ cycle separator of $\Tc \oplus G_{2}$.
  \item $\clos{C}$ uses the segment $v_2$-$x$-$v_3$ from the gadget 
    (see subfigure $(2)$ in~\cref{fig:gadget_proof}). We replace the 
    segment by a path in $G_{2}$ from $v_2$ to $v_3$ via $x$ 
    Such a path must exist as the block of $G_2$ corresponding to $H_0$ 
    is biconnected by our construction. 
    The contribution of the gadget to interior of $\clos{C}$ is $\wt(H_1)-1$,  
    which is at least as much as contribution of $G_{2}$ to interior components 
    of $\Tc$ (Since $x$ also lies in $C$, and removing $C$ would
    disconnect vertices of $H_0$). 
  \item $\clos{C}$ uses segment $v_1$-$x$-$v_2$ from the gadget 
    (see subfigures $(3a), (3b)$ in~\cref{fig:gadget_proof}). 
    It may use either one of the parallel edges between $v_1,x$. 
    We replace the segment by any path in $G_{2}$ from $v_1$ to $v_2$ (they all
    must go via $x$). 
    The contribution of the gadget to interior of $\clos{C}$ is at least 
    $\wt(H_0)-2$ (plus $\wt(H_1)-2$ possibly , depending on which of the parallel 
    $(v_1,x)$ edges is taken), whereas the contribution of $G_{2}$ to interior 
    components of $\Tc$ is at most $\wt(H_0)-2$ since both $x,v_1$ are removed.
  \item $\clos{C}$ uses segment $v_1$-$x$-$v_2$-$v_3$ 
    (see subfigure $(4)$ in~\cref{fig:gadget_proof}). 
    We replace the segment by a $v_1$-$x$ path in $G_2$, concatenated with 
    a path from $x$ to $v_3$ via $v_2$. There must exist such a path 
    in $G_2$ by the same reasoning as described above.
    Since $v_1,v_2,v_3$ all lie in $C$ in this case, the 
    contribution of $G_{2}$ is $0$ to both interior and exterior components 
    of $\Tc$.
\end{enumerate}
Therefore in all cases, we get a simple cycle by replacing the segments as
described, and the new cycle $C$ remains a balanced separator.
\end{proof}

\subsection{Appendix for~\cref{sec:DFS_btw}}\label{app:DFS_btw_app}
For brevity, we will use some shorthands in formulae like $ \exists !v \ldots$ for
`there exists exactly one $v$ such that$\ldots$', which are known to be expressible 
in $\mso$.
We define some more formulae and construct
the required expressions with explanations below.
\begin{enumerate}
  \item We first note that $\phiedge{P}{u}{v}$ which is true iff edge 
    $(u,v) \in P$, can be expressed as:
    \begin{equation*}
      \phiedge{P}{u}{v} : \exists e(\tail{e}{u}\wedge \head{e}{v} \wedge e\in P)
  \end{equation*}
  \item We define $\phiconn{P}$, which is true iff the edges in $P$ form a weakly connected 
    graph. It is well known that $\phiconn{P}$ as expressible 
    in $\mso$ so we skip the exact formula here (see for
    example~\cite{Jerabek.se}). 
  \item We define $\phisub{P_1}{P_2}$, which is true iff the set of edges of
    $P_1$ is a subset of set of edges of $P_2$, i.e., $\forall e (e\in P_1
    \Rightarrow e\in P_2)$.
  \item $\phiep{P}{v}$ can be written as:
    \begin{equation*}
      \phiep{P}{v} : \exists !u ( \phiedge{P}{u}{v}) \wedge 
                     \nexists w(\phiedge{P}{v}{w})
    \end{equation*}

     \item We can express $\phipath{P}$ as :
    \begin{equation*}
      \begin{split}
      \phipath{P} : & \phiconn{P} \wedge
      \exists !v (v \neq r \wedge \phiep{P}{v}) \wedge \\
      & \forall u ( \exists e \in P(\tail{e}{u}\lor \head{e}{u})
	    \wedge \neg \phiep{P}{u} \Rightarrow \\
      &   \exists !x,\exists !y (x\neq y \wedge \phiedge{P}{x}{u} \wedge \phiedge{P}{u}{y}))
    \end{split}
    \end{equation*}
    The last line says that for every vertex $u$ in $P$ that is not either
    of the end points $r,v$, it has exactly one in-neighbour and exactly one
    out-neighbour in $P$.
  \item Using above, we can express $\philex{P_1}{P_2}$ by saying that it holds
    iff both $\phipath{P_1},\phipath{P_2}$ hold, and there exist subpaths
    $P'_1,P'_2$ of $P_1,P_2$ respectively starting from $r$, which diverge only
    in their last edges, with $P'_1$ diverging to a vertex lesser in $\ord$ than
    the vertex $P'_2$ diverges to. We can express this as follows:
    \begin{equation*}
      \begin{split}
       \philex{P_1}{P_2} : & \phipath{P_1} \wedge \phipath{P_2} \wedge \\
       & \exists P',u_d,u_1,u_2 (u_1\neq u_2 \wedge \phipath{P'} \wedge
          \phiep{P'}{u_d} \wedge \phisub{P'}{P_1} \wedge \\
	  &  \phisub{P'}{P_2} \wedge \phiedge{P_1}{u_d}{u_1} \wedge 
	   \phiedge{P_2}{u_d}{u_2} \wedge u_1 \lord{\ord} u_2)
      \end{split}
    \end{equation*}
\end{enumerate}
This gives all expressions needed for $\phidfs{u}{v}$.

\end{document}